\documentclass[nonacm]{acmart}

\usepackage[utf8]{inputenc}
\usepackage[T1]{fontenc}
\usepackage{needspace}
\usepackage{comment} 
\usepackage{stmaryrd}
\usepackage{graphicx}
\usepackage{enumerate}
\usepackage{thm-restate}
\usepackage[overload]{textcase} 

\usepackage[most]{tcolorbox}

\usepackage{thmtools}

\usepackage{doi} 
\usepackage{tikz}

\newtheorem*{question}{Main question}
\newtheorem{remark}{Remark}
\newtheorem{claim}{Claim}

\title{What can be certified compactly?}
\subtitle{Compact local certification of MSO properties in tree-like graphs}

\author{Nicolas Bousquet}
\orcid{https://orcid.org/0000-0003-0170-0503}
\affiliation{Univ. Lyon, Université Lyon 1, LIRIS UMR CNRS 5205, F-69621, Lyon, France}

\author{Laurent Feuilloley}
\orcid{}
\affiliation{Univ. Lyon, Université Lyon 1, LIRIS UMR CNRS 5205, F-69621, Lyon, France}

\author{Théo Pierron}
\orcid{}
\affiliation{Univ. Lyon, Université Lyon 1, LIRIS UMR CNRS 5205, F-69621, Lyon, France}

\begin{abstract}
Local certification consists in assigning labels (called \emph{certificates}) to the nodes of a network to certify a property of the network or the correctness of a data structure distributed on the network. 
The verification of this certification must be local: a node typically sees only its neighbors in the network. The main measure of performance of a certification is the size of its certificates.  

In 2011, Göös and Suomela identified $\Theta(\log n)$ as a special certificate size: below this threshold little is possible, and several key properties do have certifications of this type. 
A certification with such small certificates is now called a \emph{compact local certification}, and it has become the gold standard of the area, similarly to polynomial time for centralized computing. 
A major question is then to understand which properties have $O(\log n)$ certificates, or in other words: what is the power of compact local certification? 

Recently, a series of papers have proved that several well-known network properties have compact local certifications: planarity, bounded-genus, etc. But one would like to have more general results, \emph{i.e.} meta-theorems. 
In the analogue setting of polynomial-time centralized algorithms, a very fruitful approach has been to prove that restricted types of problems can be solved in polynomial time in graphs with restricted structures. 
These problems are typically those that can be expressed in some logic, and the graph structures are whose with bounded width or depth parameters. 
We take a similar approach and prove the first meta-theorems for local certification.

More precisely, the logic we use is MSO, the most classic fragment for logics on graphs, where one can quantify on vertices and sets of vertices, and consider adjacency between vertices. 
We prove the relevance of this choice in the context of local certification by first considering properties of trees. 
On trees, we prove that MSO properties can be certified with labels of constant size, whereas the typical non-MSO property of isomorphism requires $\tilde{O}(n)$ size certificates (where $\tilde{O}$ hides polylogarithmic factors). 
We then move on to graphs of bounded treedepth, a well-known parameter that basically measures how far a graph is from a star. 
We first prove that an optimal certification for bounded treedepth uses certificates of size $\Theta(\log n)$, and then prove that in bounded treedepth graphs, every MSO property has a compact certification.

To establish our results, we use a variety of techniques, originating from model checking, tree automata theory, communication complexity, and combinatorics.

\emph{A preliminary version of this paper appeared on the arxiv under the name ``Local certification of MSO properties for bounded treedepth graphs''~\cite{BousquetFP21-MSO}.}

\end{abstract}

\begin{document}

\maketitle

\section{Introduction}
\subsection{Local certification}

In this work, we are interested in the locality of graph properties. For example, consider the property ``the graph has maximum degree three''. 
We say that this property can be checked locally, because if every node checks that it has at most three neighbors (which is a local verification), then the graph satisfies the property (which is a global statement).
Most graph properties of interest are not local. 
For example, to decide whether a graph is acyclic, or planar, the vertices would have to look arbitrarily far in the graph. Some properties can be seen as local or not, depending on the exact definition. For example, having a diameter at most $2$, is a property that can be checked locally if we consider that looking at distance $3$ is local, but not if we insist on inspecting only the neighbors of a vertex. 

As distributed computing is subject to faults and changes in the network, it is essential to be able to check properties of the network or of distributed data structures efficiently. 
Since most properties are not locally checkable, we would like to have a mechanism to circumvent this shortcoming.
Local certification is such a mechanism, in the sense  that it allows to check any graph property locally. 
For a given property, a local certification is described by a certificate assignment and a verification algorithm: each node receives a certificate, reads the certificates of its neighbors and then runs a verification algorithm. This algorithm decides whether the node accepts or rejects the certification. 
If the graph satisfies the property, then there should be a certificate assignment such that all the nodes accept. Otherwise, in each assignment, there must be at least one node that rejects.

In recent years, the area of local certification has attracted a lot of attention, and we refer to ~\cite{FeuilloleyF16} and~\cite{Feuilloley21} for respectively a complexity-theory oriented survey, and an introduction to the area.

\subsection{Understanding the power of compact local certification}

It is known that any property can be certified with $O(n^2)$ bits certificates, where $n$ is the total number of vertices. 
This is because one can simply give the full description of the graph to every node, which can then check that the property holds in the graph described, and that the graph description is correct locally, and identical between neighbors. 
This $O(n^2)$ size is extremely large, and the main goal of the study of local certification is to minimize the size of the certificate, expressed as a number of bits per vertex, as a function of $n$.
In addition to the optimization motivation originating from distributed self-stabilizing algorithms, establishing the minimum certificate size also has a more theoretical appeal. Indeed, the optimal certification size of a property can be seen as a measure of its locality: the smaller the labels, the less global information we need to allow local verification, the more local the property.

In~\cite{GoosS16}, Göös and Suomela identified that $\Theta(\log n)$ was a special certificate size.
Indeed, on the one hand, for most properties, one cannot hope to go below this size per node. For example, certifying acyclicity requires $\Omega(\log n)$~\cite{KormanKP10, GoosS16}.
On the other hand, once we have a logarithmic number of bits, we can encode identifiers, and distances, which is enough for spanning trees, an essential tool in certification. 
A certification with $\Theta(\log n)$ bits is now called a \emph{compact local certification}, and it has become the gold standard of the area.
Recently, planarity and more generally embeddability on bounded-genus surfaces, and $H$-minor-freeness for graphs $H$ of size at most $4$, have been proved to have such compact certifications~\cite{FeuilloleyFMRRT21, FeuilloleyFMRRT21-b, EsperetL22, BousquetFP21}.

Unfortunately, not every property has a compact certification.
For example, having a non-trivial automorphism or not being $3$-colorable are properties that cannot be certified with less than $\Omega(n^2)$ bits~\cite{GoosS16}.
Even surprisingly simple properties, such as having diameter at most $2$, cannot be certified with a sublinear number of bits per vertex (up to logarithmic factors), if we only allow the local verification to be at distance~one~\cite{Censor-HillelPP20}. 
This raises the following question:

\begin{question}
\label{question:compact}
What are the graph properties that admit a compact local certification?
\end{question}

\section{Our approach, results, and techniques}
\label{sec:approach}

\subsection{A systematic model checking approach}

As mentioned above, many specific graph properties such as planarity or small-diameter have been studied in the context of local certification. 
In this paper, we take a more systematic approach, inspired by model checking, by considering classes of graph properties. 
We are interested in establishing theorems of the form: ``all the properties that can be expressed in some formalism $X$ have a compact certification''. 

In this paper, we will consider properties that can be expressed by sentences from monadic second order logic (MSO), just like in Courcelle's theorem. 
These are formed from atomic predicates that test equality or adjacency of vertices 
and allowing boolean operations and quantifications on vertices, edges, and sets of vertices or edges. Now, certifying a given property consists in certifying that a graph is a positive instance of the so-called graph model checking problem for the corresponding sentence $\varphi$: 
\begin{itemize} 
    \item Input: A graph $G$
    \item Output: Yes, if and only if, $G$ satisfies $\varphi$.
\end{itemize}

\subsection{The generic case}

Let us first discuss what such a meta-theorem must look like when we do not restrict the class of graphs we consider. As we already mentioned,  graphs of diameter at most~$2$ cannot be certified with sublinear certificates~\cite{Censor-HillelPP20}. 
This can be expressed with the following sentence:
$$\forall x\forall y (x=y\vee x-y \vee \exists z (x-z \land z-y))$$
This sentence is very simple: it is a first order sentence (a special case of MSO), it has quantifier depth three and there is only one quantifier alternation (two standard complexity measures for FO sentences which respectively counts the maximum number of nested quantifiers and the number of alternations between blocks of existential and universal quantifiers). 
Therefore, there exists very simple first order logic sentences which cannot be certified efficiently, hence  there is no room for a generic $O(\log n)$ result. 

Note that if we allowed the vertices to see at a larger (but still constant) distance in the graph, then we could verify diameter~$2$ without certificates.
In order to prevent such phenomenon, and because it is more relevant in terms of message complexity, in the whole paper, the radius of the views of the vertices is fixed to~$1$ (in other words, a node can read the IDs and the certificates of all its neighbors, but cannot see which edges are incident to these vertices). 
We discuss that aspect in more detail in Appendix~\ref{app:discussion-radius}. 

Another example is given by triangle-freeness, which can be expressed by the following sentence:
$$\forall x\forall y\forall z \neg (x-y\land y-z \land x-z)$$
This sentence also has rank $3$ and no quantifier alternation. Proposition 5 of \cite{CrescenziFP19} proves that certifying that a graph is triangle-free requires~$\Omega(n/e^{O(\sqrt{n})})$ bits, via reduction to multi-party communication complexity inspired by~\cite{DruckerKO13}. 

The only possible way to simplify the sentences would consist in only having at most two nested quantifiers or not authorizing universal quantifiers. In these cases, the following holds:

\begin{restatable}{lemma}{lemFragments}
\label{lem:existsFO}
FO sentences with quantifier depth at most 2 can be certified with $O(\log n)$ bits. 
Existential FO sentences (\emph{i.e.} whose prenex normal form has only existential quantifiers) can be certified with $O(\log n)$ bits.
\end{restatable}

For the FO sentences with quantifiers of depth at most $2$, we can prove that the only interesting properties that can be expressed are a vertex being dominant (adjacent to all other vertices) or the graph being a clique. 
These are easy to certify with $O(\log n)$ bits, \emph{c.f.} the full proofs in Appendix~\ref{sec:fragment}.

\subsection{The case of trees}

Lemma~\ref{lem:existsFO} characterizes the quite restricted sentences that can be certified with $O(\log n)$ bits for general graphs.
The classic approach in centralized computing is then to restrict the class of the graphs considered. 
This is also relevant here: for example, certifying some given diameter is easier if we restrict the graphs to trees. 
Indeed, in this case we can use a spanning tree to point to a central vertex (or edge), that becomes a root (or root-edge), and keep at every vertex both its distance to the root and the depth of its subtree.
This certification can be checked by simple distance comparisons, and it uses $O(\log n)$ bits. The first of our main results is that we can actually get a better bound (constant certificates) for all MSO properties on trees.

\begin{restatable}{theorem}{thmMSOTrees}
\label{thm:MSO-trees}
Any MSO formula can be certified on trees with certificates of size $O(1)$.
\end{restatable}

One can wonder if we can extend this statement to a significantly wider logic.
We answer by the negative by proving that some typical non-MSO properties cannot be certified with certificates of sublinear sizes even on trees of bounded depth. 

\begin{restatable}{theorem}{thmIsomorphismTrees}
\label{thm:isomorphism-trees}
Certifying the trees that have an automorphism without fixed-point requires certificates of size $\tilde{\Omega}(n)$ (where $\tilde{\Omega}$ hides polylogarithmic factors), even if we restrict to trees of bounded depth.
\end{restatable}

\subsection{The case of bounded treedepth graphs}

In centralized model checking, a classic meta-theorem of Courcelle~\cite{Courcelle90} establishes that all the problems expressible in MSO can be solved in polynomial-time in graphs of bounded treewidth.
Motivated by the unavoidable non-elementary dependence in the formula in Courcelle's theorem~\cite{Frick04}, Gajarsk\'y and Hlin\v en\'y~\cite{GajarskyH15} designed a linear-time FPT algorithm for MSO-model checking with elementary dependency in the sentence, by paying the price of considering a smaller class of graphs, namely graphs of bounded treedepth. Their result is essentially the best possible as shown soon after in~\cite{Lampis13}. 

One can wonder if some Courcelle-like result holds for certification. Namely, is it possible to certify any MSO-formula on graphs of bounded treewidth with certificates of size $O(\log n)$? Prior to our work, it was not known whether graphs of fixed width can be certified with logarithmic size certificates. Proving such a statement is a preliminary condition for MSO-certification, since certifying a property on a graph class we cannot certify may lead to aberrations.

We prove that one can locally check that a graph has treedepth at most $t$ with logarithmic-size certificates. 

\begin{restatable}{theorem}{thmCertifyTreedepth}
\label{thm:certify-treedepth}
We can certify that a graph has treedepth at most~$t$ with $O(t \log n)$ bits.
\end{restatable}

We also show that Theorem~\ref{thm:certify-treedepth} is optimal, in the sense that certifying treedepth at most $k$ requires $\Omega(\log n)$ bits, even for small $k$. 

\begin{restatable}{theorem}{thmLowerBoundTreedepth}
\label{thm:lower-bound-treedepth}
Certifying that the treedepth of the graph is at most $k$ requires $\Omega(\log n)$ bits, for any $k\geq 5$. 
\end{restatable}

This result contrasts with the fact that certifying trees of depth $k$ can be done with $O(\log k)$ bits (thus independent of $n$), by simply encoding distances to the root.

The next problem in line is then MSO-model checking for graphs of bounded treedepth. In such classes, it happens that MSO and FO have the same expressive power~\cite{ElberfeldGT16}: for every $t$ and every MSO sentence, there exists a FO sentence satisfied by the same graphs of treedepth at most $t$. 

\begin{restatable}{theorem}{thmMain}
\label{thm:main}
Every FO (and hence MSO) sentence $\varphi$ can be locally certified with $O(t\log n+ f(t,\varphi))$-bit certificates on graphs of treedepth at most $t$. 
\end{restatable}

This result, as well as Theorem~\ref{thm:MSO-trees}, holds for MSO properties about the structure of the graphs, but our techniques also work for graphs with constant-size inputs, in the spirit of locally checkable labelings~\cite{NaorS95}.

Inspired by our results and techniques, Fraigniaud, Montealegre, Rapaport, and Todinca, very recently proved that it is possible to certify MSO properties in bounded treewidth graphs, with certificates of size $\Theta(\log^2n)$~\cite{FraigniaudMRT21}. 
Replacing treedepth by treewidth is very interesting, as the second parameter is more general and well-known, but it comes at the cost of certificates of size $\Theta(\log^2n)$, hence not a compact certification \emph{per se}. 
It is a fascinating question whether this is optimal or can be reduced down to $O(\log n)$.

Theorem~\ref{thm:main} has an interesting corollary for the certification of graphs with forbidden minors. 
An important open question in the field of local certification is to establish whether all the graph classes defined by a set of forbidden minors have a compact certification (e.g. Open problem 4 in \cite{Feuilloley21}).
Note that this question generalizes the results about planarity and bounded-genus graphs of \cite{FeuilloleyFMRRT21, FeuilloleyFMRRT21-b, EsperetL22}. 
Very recently, Bousquet, Feuilloley and Pierron proved that the answer is positive for all minors of size at most $4$~\cite{BousquetFP21}, but the question is still wide open for general minors.
Theorem~\ref{thm:main} leads to the following result, where $P_t$ and $C_t$ are respectively the path and the cycle of length $t$.

\begin{restatable}{corollary}{coroMinors}
\label{coro:minors}
For all $t$, $P_t$-minor-free graphs and $C_t$-minor-free graphs can be certified with $O(\log n)$-bit certificates.
\end{restatable}

Still related to the certification of minors, 
Esperet and Norin~\cite{EsperetN22} (generalizing a result by Elek~\cite{Elek20}) proved very recently that certifying that a graph belongs to a minor-closed class or is far from it (in the sense of the edit distance, as in property testing) can be done with constant size certificate. 
Using our certification of bounded treedepth, they generalize this result to all monotone properties of minor-closed classes, with $O(\log n)$-size certificates.

Let us finish this overview, by mentioning a related line of research. A recent series of papers have characterized diverse logics on graphs by various models of distributed local computation, in a similar way as descriptive complexity in centralized computing~\cite{Immerman99}. 
In this area, a paper that is especially relevant to us is~\cite{Reiter15}, which proves that MSO logic on graphs is equivalent to a model called alternating distributed graph automata.
These are actually quite different from our model, with several provers, more constrained local computation, and more general output functions. We describe this model and discuss the differences in more details in Appendix~\ref{app:discussion-DGA}.

\subsection{A glimpse of our techniques and the organization of the paper}

We use a variety of techniques to prove our results, and except for a section of preliminaries (Section~\ref{sec:preliminaries}), each upcoming section of this paper corresponds to one technique. First, we show how to prove the constant size MSO certification in trees (Theorem~\ref{thm:MSO-trees}) by seeing the certificates as a state labeling by the right type of tree automata, and then using the known logic-automata correspondence to derive our result. We will discuss in the appendix how this automata view can be an inspiration to generalize locally checkable languages (LCLs)~\cite{NaorS95} beyond bounded degree graphs.  

The proof of the certification of bounded treedepth (Theorem~\ref{thm:certify-treedepth}) is in Section~\ref{sec:treedepth-certification}, and uses spanning tree certification along with an analysis of interplay between ancestors in the decompositions and the separators in the graph. 
Given this certification, we certify MSO properties (Theorem~\ref{thm:main}) via kernelization. In more details, we show that for any graph there exists a kernel, that is, a graph that satisfies the exact  same set of MSO properties, whose size only depends on the formula and on the treedepth (and in particular not in the size of the original graph). We show that this kernel can be certified locally, which is enough for our purpose, as we can finish by describing the full kernel to all nodes, and let them check the MSO property at hand.  

Finally, in Section~\ref{sec:lower-bounds}, we prove our two lower bounds (Theorem~\ref{thm:isomorphism-trees} and~\ref{thm:lower-bound-treedepth}) by reduction from two-party non-deterministic communication complexity.

To our knowledge, it is the first time that automata tools, kernelization, and reductions from communication complexity for the $\Theta(\log n)$ regimes, are used in local certification.

\section{Preliminaries}
\label{sec:preliminaries}

All the graphs considered in this paper are connected, loopless and non-empty.

\subsection{Treedepth}
Treedepth was introduced by Ne\v{s}et\v{r}il and Ossona de Mendez in~\cite{NesetrilM06} as a graph parameter inducing a class where model checking is more efficient. In the last ten years, this graph parameter received considerable attention (see \cite{NesetrilM12} for a book chapter about this parameter). Treedepth is related to other important width parameters in graphs. In particular, it is an upper bound on the pathwidth, which is essential in the study of minors~\cite{RobertsonS83} and interval graphs~\cite{Bodlaender98}.

Let $T$ be a rooted tree. A vertex $u$ is an \emph{ancestor} of $v$ in $T$, if $u$ is on the path between $v$ and the root. We say that $v$ is a \emph{descendant} of $u$ if $u$ is an ancestor of $v$.

\begin{definition}[\cite{NesetrilM06}]
The \emph{treedepth} of a graph $G$ is the minimum height of a forest $F$ on the same vertex set as $G$, such that for every edge $(u,v)$ of the graph $G$, $u$ is an ancestor or a descendant of $v$ in the forest. 
\end{definition}

Since in our setting $G$ is connected, $F$ is necessarily a tree, called an \emph{elimination tree}. 
In a more logic-oriented perspective, it is called a \emph{model} of the graph. If the tree has depth at most~$k$, it is a \emph{$k$-model of $G$} (see Figure~\ref{fig:treedepth}).
Note that there might be several elimination trees. 

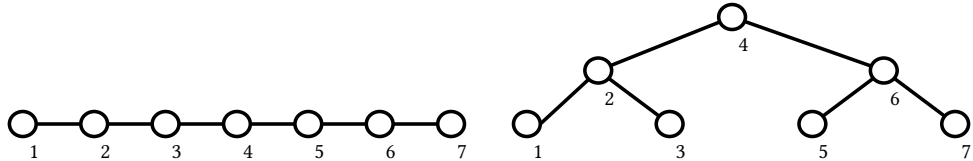
\begin{figure}[!h]
    \centering
    \begin{tabular}{cc}
    \scalebox{0.9}{
    \tikzset{every picture/.style={line width=0.75pt}} 

\begin{tikzpicture}[x=0.75pt,y=0.75pt,yscale=-1,xscale=1]

\draw [line width=1.5]    (100,137) -- (140,137) ;
\draw [line width=1.5]    (132.5,137) -- (172.5,137) ;
\draw [line width=1.5]    (172.5,137) -- (212.5,137) ;
\draw [line width=1.5]    (212.5,137) -- (252.5,137) ;
\draw [line width=1.5]    (252.5,137) -- (292.5,137) ;
\draw [line width=1.5]    (292.5,137) -- (332.5,137) ;
\draw  [fill={rgb, 255:red, 255; green, 255; blue, 255 }  ,fill opacity=1 ][line width=1.5]  (85,137) .. controls (85,133.13) and (88.36,130) .. (92.5,130) .. controls (96.64,130) and (100,133.13) .. (100,137) .. controls (100,140.87) and (96.64,144) .. (92.5,144) .. controls (88.36,144) and (85,140.87) .. (85,137) -- cycle ;
\draw  [fill={rgb, 255:red, 255; green, 255; blue, 255 }  ,fill opacity=1 ][line width=1.5]  (125,137) .. controls (125,133.13) and (128.36,130) .. (132.5,130) .. controls (136.64,130) and (140,133.13) .. (140,137) .. controls (140,140.87) and (136.64,144) .. (132.5,144) .. controls (128.36,144) and (125,140.87) .. (125,137) -- cycle ;
\draw  [fill={rgb, 255:red, 255; green, 255; blue, 255 }  ,fill opacity=1 ][line width=1.5]  (165,137) .. controls (165,133.13) and (168.36,130) .. (172.5,130) .. controls (176.64,130) and (180,133.13) .. (180,137) .. controls (180,140.87) and (176.64,144) .. (172.5,144) .. controls (168.36,144) and (165,140.87) .. (165,137) -- cycle ;
\draw  [fill={rgb, 255:red, 255; green, 255; blue, 255 }  ,fill opacity=1 ][line width=1.5]  (205,137) .. controls (205,133.13) and (208.36,130) .. (212.5,130) .. controls (216.64,130) and (220,133.13) .. (220,137) .. controls (220,140.87) and (216.64,144) .. (212.5,144) .. controls (208.36,144) and (205,140.87) .. (205,137) -- cycle ;
\draw  [fill={rgb, 255:red, 255; green, 255; blue, 255 }  ,fill opacity=1 ][line width=1.5]  (245,137) .. controls (245,133.13) and (248.36,130) .. (252.5,130) .. controls (256.64,130) and (260,133.13) .. (260,137) .. controls (260,140.87) and (256.64,144) .. (252.5,144) .. controls (248.36,144) and (245,140.87) .. (245,137) -- cycle ;
\draw  [fill={rgb, 255:red, 255; green, 255; blue, 255 }  ,fill opacity=1 ][line width=1.5]  (285,137) .. controls (285,133.13) and (288.36,130) .. (292.5,130) .. controls (296.64,130) and (300,133.13) .. (300,137) .. controls (300,140.87) and (296.64,144) .. (292.5,144) .. controls (288.36,144) and (285,140.87) .. (285,137) -- cycle ;
\draw  [fill={rgb, 255:red, 255; green, 255; blue, 255 }  ,fill opacity=1 ][line width=1.5]  (325,137) .. controls (325,133.13) and (328.36,130) .. (332.5,130) .. controls (336.64,130) and (340,133.13) .. (340,137) .. controls (340,140.87) and (336.64,144) .. (332.5,144) .. controls (328.36,144) and (325,140.87) .. (325,137) -- cycle ;

\draw (94.5,147.4) node [anchor=north west][inner sep=0.75pt]    {$1$};
\draw (134.5,147.4) node [anchor=north west][inner sep=0.75pt]    {$2$};
\draw (174.5,147.4) node [anchor=north west][inner sep=0.75pt]    {$3$};
\draw (214.5,147.4) node [anchor=north west][inner sep=0.75pt]    {$4$};
\draw (254.5,147.4) node [anchor=north west][inner sep=0.75pt]    {$5$};
\draw (294.5,147.4) node [anchor=north west][inner sep=0.75pt]    {$6$};
\draw (334.5,147.4) node [anchor=north west][inner sep=0.75pt]    {$7$};

\end{tikzpicture}}
    &
    \scalebox{0.9}{
    \tikzset{every picture/.style={line width=0.75pt}} 

\begin{tikzpicture}[x=0.75pt,y=0.75pt,yscale=-1,xscale=1]

\draw [line width=1.5]    (100,137) -- (132.5,107) ;
\draw [line width=1.5]    (132.5,107) -- (172.5,137) ;
\draw [line width=1.5]    (132.5,107) -- (207.5,77) ;
\draw [line width=1.5]    (292.5,107) -- (252.5,137) ;
\draw [line width=1.5]    (292.5,107) -- (207.5,77) ;
\draw [line width=1.5]    (292.5,107) -- (332.5,137) ;
\draw  [fill={rgb, 255:red, 255; green, 255; blue, 255 }  ,fill opacity=1 ][line width=1.5]  (85,137) .. controls (85,133.13) and (88.36,130) .. (92.5,130) .. controls (96.64,130) and (100,133.13) .. (100,137) .. controls (100,140.87) and (96.64,144) .. (92.5,144) .. controls (88.36,144) and (85,140.87) .. (85,137) -- cycle ;
\draw  [fill={rgb, 255:red, 255; green, 255; blue, 255 }  ,fill opacity=1 ][line width=1.5]  (125,107) .. controls (125,103.13) and (128.36,100) .. (132.5,100) .. controls (136.64,100) and (140,103.13) .. (140,107) .. controls (140,110.87) and (136.64,114) .. (132.5,114) .. controls (128.36,114) and (125,110.87) .. (125,107) -- cycle ;
\draw  [fill={rgb, 255:red, 255; green, 255; blue, 255 }  ,fill opacity=1 ][line width=1.5]  (165,137) .. controls (165,133.13) and (168.36,130) .. (172.5,130) .. controls (176.64,130) and (180,133.13) .. (180,137) .. controls (180,140.87) and (176.64,144) .. (172.5,144) .. controls (168.36,144) and (165,140.87) .. (165,137) -- cycle ;
\draw  [fill={rgb, 255:red, 255; green, 255; blue, 255 }  ,fill opacity=1 ][line width=1.5]  (200,77) .. controls (200,73.13) and (203.36,70) .. (207.5,70) .. controls (211.64,70) and (215,73.13) .. (215,77) .. controls (215,80.87) and (211.64,84) .. (207.5,84) .. controls (203.36,84) and (200,80.87) .. (200,77) -- cycle ;
\draw  [fill={rgb, 255:red, 255; green, 255; blue, 255 }  ,fill opacity=1 ][line width=1.5]  (245,137) .. controls (245,133.13) and (248.36,130) .. (252.5,130) .. controls (256.64,130) and (260,133.13) .. (260,137) .. controls (260,140.87) and (256.64,144) .. (252.5,144) .. controls (248.36,144) and (245,140.87) .. (245,137) -- cycle ;
\draw  [fill={rgb, 255:red, 255; green, 255; blue, 255 }  ,fill opacity=1 ][line width=1.5]  (285,107) .. controls (285,103.13) and (288.36,100) .. (292.5,100) .. controls (296.64,100) and (300,103.13) .. (300,107) .. controls (300,110.87) and (296.64,114) .. (292.5,114) .. controls (288.36,114) and (285,110.87) .. (285,107) -- cycle ;
\draw  [fill={rgb, 255:red, 255; green, 255; blue, 255 }  ,fill opacity=1 ][line width=1.5]  (325,137) .. controls (325,133.13) and (328.36,130) .. (332.5,130) .. controls (336.64,130) and (340,133.13) .. (340,137) .. controls (340,140.87) and (336.64,144) .. (332.5,144) .. controls (328.36,144) and (325,140.87) .. (325,137) -- cycle ;

\draw (94.5,147.4) node [anchor=north west][inner sep=0.75pt]    {$1$};
\draw (134.5,117.4) node [anchor=north west][inner sep=0.75pt]    {$2$};
\draw (174.5,147.4) node [anchor=north west][inner sep=0.75pt]    {$3$};
\draw (209.5,87.4) node [anchor=north west][inner sep=0.75pt]    {$4$};
\draw (254.5,147.4) node [anchor=north west][inner sep=0.75pt]    {$5$};
\draw (294.5,117.4) node [anchor=north west][inner sep=0.75pt]    {$6$};
\draw (334.5,147.4) node [anchor=north west][inner sep=0.75pt]    {$7$};

\end{tikzpicture}}
    \end{tabular}
    \caption{An example of an elimination tree. On the left the graph $G$, that is a path on seven vertices, and on the right an elimination tree $T$ of this graph. Since this tree has depth $2$, the path has treedepth at most $2$, and this is actually optimal.}
    \label{fig:treedepth}
\end{figure}

Let us fix an elimination tree. A \emph{vertex of $G$ has depth $d$}, if it has depth $d$ in the elimination tree. For any vertex $v$, let $G_v$ be the subgraph of $G$ induced by the vertices in the subtree of $T$ rooted in~$v$. Note that, for the root~$r$, $G_r=G$. Now, a model $T$ of $G$ is \emph{coherent} if, for every vertex $v$, the vertices of the subforest rooted in $v$ form a connected component in~$G$. In other words, for every child $w$ of $v$, there exists a vertex $x$ of the subtree rooted in $w$ that is connected to~$v$.

We have the following simple result, that we prove in Appendix~\ref{app:remark-coherent} for completeness.

\begin{restatable}{remark}{rkCoherent}
\label{rk:coherent}
Let $T$ be a coherent $d$-model of a connected graph $G$ and $u$ be a vertex of $G$. Then $G_u$ induces a connected subgraph.
\end{restatable}

\subsection{FO and MSO logics}
\label{subsec:preliminaries-logics}

Graphs can be seen as relational structures on which properties can be expressed using logical sentences. The most natural formalism considers a binary predicate that tests the adjacency between two vertices. Allowing standard boolean operations and \emph{quantification on vertices}, we obtain the \emph{first-order logic} (FO for short) on graphs. Formally, a FO formula is defined by the following grammar:
$$x=y \mid x-y \mid \neg F \mid F\wedge F \mid F\vee F  \mid \forall x F \mid \exists x F$$
where $x,y$ lie in a fixed set of variables. Except for $x-y$, which denotes the fact that $x$ and $y$ are adjacent, the semantic is the classic one. 
Given a FO sentence $F$ (i.e. a formula where each variable falls under the scope of a corresponding quantifier) and a graph $G$, we write $G\vDash F$ when the graph $G$ satisfies the sentence $F$, which is defined in the natural way. 

MSO logic is an enrichment of FO, where we allow \emph{quantification on sets of vertices}\footnote{Sometimes MSO sentences are also allowed to quantify on set of edges. We do not discuss the matter any further since for our purposes (i.e. on trees or bounded treedepth graphs), it is known that quantifying on edges does not increase the expressive power of the sentences.}, usually denoted by capital variables, and we add the membership predicate $x\in X$. We skip the details here since for bounded treedepth graphs, it is known that FO and MSO have the same expressive power.

\begin{theorem}[\cite{Grohe17}]
For every integer $d$ and MSO sentence $\varphi$, there exists a FO sentence $\psi$ such that $\varphi$ and $\psi$ are satisfied by the same set of graphs of treedepth at most $d$. 
\end{theorem}

In Section~\ref{sec:MSO-certif-bounded-treedepth}, we are looking for a kernelization result for the model checking problem, where the kernel is checkable with small certificates. In particular, given a sentence $\varphi$ and a graph $G$, we have to prove that the graph $H$ output by our kernelization algorithm satisfies $\varphi$ if and only if so does $G$. We actually show a stronger result, namely that for every integer $k$ and every graph $G$, there exists a graph $H_k$ satisfying the same set of sentences with at most $k$ nested quantifiers as $G$. In that case, we write $G\simeq_k H_k$. This yields the required result when $k$ is quantifier depth of $\varphi$. 

The canonical tool to prove equivalence between structures is the so-called \emph{Ehrenfeucht-Fraïssé game}. This game takes place between two players, Spoiler and Duplicator. The arena is given by two structures (here, graphs) and a number $k$ of rounds. At each turn, Spoiler chooses a vertex in one of the graphs, and Duplicator has to answer by picking a vertex in the other graph. 
Let the positions played in the first (resp. second) graph at turn $i$ be $u_1,\ldots,u_i$ (resp. $v_1,\ldots,v_i$). Spoiler wins at turn $i$ if the mapping $u_j\mapsto v_j$ is not an isomorphism between the subgraphs induced by $\{u_1,\ldots,u_i\}$ and $\{v_1,\ldots,v_i\}$. If Spoiler does not win before the end of the $k$-th turn, then Duplicator wins.
The main result about this game is the following, which relates winning strategies with equivalent structures for $\simeq_k$. 

\begin{theorem}
\label{thm:EF}
Let $G,H$ be two graphs and $k$ be an integer. Duplicator has a winning strategy in the $k$-round Ehrenfeucht-Fraïssé game on $(G,H)$ if and only if $G\simeq_k H$. 
\end{theorem}

See \cite{Thomas93} for a survey on Ehrenfeucht-Fraïssé games and its applications in computer science.

\subsection{Local certification: definitions and basic techniques}
\label{subsec:basic-certif}

We assume that the vertices of the graph are equipped with unique identifiers, also called IDs, in a polynomial range $[1, n^k]$ ($k$ being a constant). Note that an ID can be written on $O(\log n)$ bits.

In this paper, a local certification is described by a local verification algorithm, which is an algorithm that takes as input the identifiers and the labels of a node and of its neighbors, and outputs a binary decision, usually called \emph{accept} or \emph{reject}. A local certification of a property is a local verification algorithm such that:
\begin{itemize}
    \item If the graph satisfies the property, then there exists a label assignment, such that the local verification algorithm accepts at every vertex.
    \item If the graph does not satisfy the property, then for every label assignment, there exists at least one vertex that rejects. 
\end{itemize}

A graph that satisfies the property is a \emph{yes-instance}, and a graph that does not satisfy the property is a \emph{no-instance}. 
The labels are called \emph{certificates}. It is equivalent to consider that there is an entity, called the \emph{prover}, assigning the labels (a kind of external oracle). The size $f(n)$ of a certification is the size of its largest label for graphs of size $n$. The certification size of a property or a set of properties is the (asymptotic) minimum size of a local certification. 

A standard tool for local certification is spanning trees that have a compact certification.

\begin{proposition}
One can locally encode and certify a spanning tree with $O(\log n)$ bits. The number of vertices in the graph can also be certified with $O(\log n)$ bits.
\end{proposition}

The idea of the certification of the spanning tree is to root the tree, and then to label the vertices with the distance to the root (to ensure acyclicity) and the ID of the root (to ensure connectivity). To certify the number of vertices, one also labels the vertices with the number of nodes in their subtrees. 
We refer to the tutorial \cite{Feuilloley21}, for intuitions, proofs, and history of these tools.

\section{MSO certification on trees via tree automata}
\label{sec:MSO-trees-automata}

\thmMSOTrees*

The full formal proof of Theorem~\ref{thm:MSO-trees} is deferred to the Appendix~\ref{app:proof-MSO-trees}, but we discuss the intuition here.
The idea of the proof is to adapt results from the tree automata literature. Let us give some intuition with classic (word) automata. 
Consider a word as a directed path whose edges are labeled with letters, then this word is recognized by an automaton if we can label the vertices with states of the automaton, in such a way that each triplet $(u, (u,v), v)$ (where $u$ and $v$ are adjacent vertices) has a labeling ($q,\ell, q'$) (where $q$ and $q'$ are states, and $\ell$ is a letter) that is a proper transition, and the first and last vertices are labeled with initial and final states respectively. 
Now to certify that a word is recognized by an automaton, we can label every node with its state in an accepting run, and the verification can be done locally. 
Finally, Büchi-Elgot-Trakhtenbrot theorem states that MSO properties are exactly the ones that are recognized by a regular automaton, thus we get Theorem~\ref{thm:MSO-trees} in the case of directed paths.
The automata point of view (without the relation to logics) has been used before to understand the complexity of locally checkable labelings on cycles and paths, see in particular in~\cite{ChangSS21}.

Now, a tree automaton is the analogue of a regular automaton for rooted trees. 
In particular, the transitions specify states for a vertex and its children. 
Again, there is a nice relation with MSO:  MSO logic on trees is exactly the set of languages recognized by tree automata~\cite{ThatcherW68}. 
Therefore, the same labeling-by-states strategy basically works, but there are some technicalities. 
Indeed, the result of~\cite{ThatcherW68}, is for rooted trees with bounded degree and with an order on the children of each node; And the properties expressible in MSO in this type of trees are a bit different from the ones in our unrooted, unordered trees with unbounded degrees. 
But we can get the result by describing a root in the certificates, and using less classical results for other types of tree automata, adapted to our type of trees~\cite{BonevaT05}. 

Interestingly, the tree automata that capture MSO properties on trees can be described as checking that the multiset of states of the neighbors satisfies some simple inequalities. 
We discuss in Appendix~\ref{app:discussion-LCL} how this provides interesting directions to generalize the classic and well-understood setting of locally checkable labelings (LCLs)~\cite{NaorS95}. 

\section{Treedepth certification via ancestors lists}
\label{sec:treedepth-certification}

This section is devoted to the proof of the following theorem.

\thmCertifyTreedepth*

Let $v$ be a vertex, and $w$ be its parent in the tree, we define an \emph{exit vertex of $v$} as a vertex $u$ of $G_v$ connected to $w$. Note that such a vertex must exist, if the model is coherent. 

We now describe a certification.
In a \emph{yes}-instance, the prover finds a coherent elimination tree of depth at most $t$, and assigns the labels in the following way.

\begin{itemize}
    \item Every vertex $v$ is given the list of the identifiers of its ancestors, from its own identifier to the identifier of the root.
    \item For every vertex $v$, except the root, the prover describes and certifies a spanning tree of $G_v$, pointing to the exit vertex of $v$. (See Subsection~\ref{subsec:basic-certif} for the certification of spanning trees.) The vertices of the spanning tree are also given the depth $k$ of $v$ in the elimination tree.
\end{itemize}

Note that the length of the lists is upper bounded by $t$, and that every vertex holds a piece of spanning tree certification only for the vertices of its list, therefore the certificates are on $O(t \log n)$ bits. Now, the local verification algorithm is the following. For every vertex $v$ with a list $L$ of length~$d+1$, check that:

\begin{enumerate}
    \item \label{item:at-most-t} $d\leq t$, and $L$ starts with the identifier of the vertex, and ends with the same identifier as in the lists of its neighbors in the graph.
    \item \label{item:edges-ancestors} The neighbors in $G$ have lists that are suffixes or extensions by prefix of $L$.
    \item \label{item:ST-for-each-depth} There are $d$ spanning trees described in the certificates.
    \item \label{item:ST-verif} For every $k\leq d$, for the spanning trees associated with depth $d$: 
    \begin{itemize}
        \item The tree certification is locally correct.
        \item The neighbors in the tree have lists with the same $(k+1)$-suffix.
        \item If the vertex is the root, then it has a neighbor whose list is the $k$-suffix of its own list. 
    \end{itemize}
\end{enumerate}

It is easy to check that on \emph{yes}-instances the verification goes through.
Now, consider an instance where all vertices accept. We shall prove that then we can define a forest, such that the lists of identifiers given to the nodes are indeed the identifiers of the ancestors in this forest. Once this is done, the fact that Steps~\ref{item:at-most-t} and~\ref{item:edges-ancestors} accept implies that the forest is a tree of the announced depth, and is a model of the graph. Let us first prove the following claim:

\begin{claim}
\label{clm:one-step}
For every vertex $u$, with a list $L$ of size at least two, there exists another vertex $v$ in the graph whose list is the same as $L$ but without the first element. 
\end{claim}

Consider a vertex $u$ like in Claim~\ref{clm:one-step}, at some depth $d$. If all vertices accept, then this vertex is has a spanning tree corresponding to depth $d$ (by Step~\ref{item:ST-for-each-depth}), where all vertices have the same $(d+1)$-suffix, and the root of this tree has a neighbor whose list is $L$, without the first identifier, by Step~\ref{item:ST-verif}. This vertex is the $v$ of the claim. 

The claim implies that the whole tree structure is correct. Indeed, if we take the vertex set of $G$, and add a pointer from every vertex $u$ to its associated vertex $v$ (with the notations of the claim), then the set of pointers must form a forest. In particular, there cannot be cycles, because the size of the list is decremented at each step. Also, if the ancestors are consistent at every node, then they are consistent globally. This finishes the proof of Theorem~\ref{thm:certify-treedepth}.

\section{MSO/FO certification in bounded treedepth graphs via kernelization} 
\label{sec:MSO-certif-bounded-treedepth}

In this section, we prove the following theorem.

\thmMain*

The proof is based on a kernelization result: we show that for every integer $t$ and $k$, for every graph of treedepth $t$, we can associate a graph, called \emph{a kernel}, such that (1) it satisfies the same FO formulas with quantifier depth at most $k$, and (2) it has a size that is independent of $n$ (that is, depends only on $t$ and $k$). The idea is then to locally describe and certify this kernel, and to let the vertices check that the kernel satisfies the formula.

Actually, such a kernel always exists, even without the treedepth assumption.
Indeed, since we have a bounded number of formulas of quantifier depth at most $k$ (up to semantic equivalence), we have a bounded number of equivalent classes of graphs for $\simeq_k$. We can associate to each class the smallest graph of the class, whose size is indeed bounded by a function of only $k$.
However, this definition of $H_k$ is not constructive, which makes it impossible to manipulate for certification. 
We note that a constructive kernelization result already exists for graphs of bounded shrubdepth~\cite{GajarskyH15}, which implies bounded treedepth. 
We however cannot use this result either, because we cannot locally certify the kernel of ~\cite{GajarskyH15}. 
Hence, we need to design our own certifiable kernel.
Incidentally, certifying bounded shrubdepth and the associated model checking problem are interesting open questions.

\subsection{Description of the kernel}
\label{subsec:kernel}


Let $G$ be a graph of treedepth at most $t$, and let $k$ be an integer. Let $\mathcal{T}$ be a $t$-model of $G$. Let $v$ be a vertex of depth $i$ in the decomposition. 
We define the \emph{ancestor vector of $v$} as the $\{0,1\}$-vector of size $i$, where the $j$-th coordinate is 1, if and only if, $v$ is connected in $G$ to its ancestor at depth $j$. 

We can now define the \emph{type of a vertex $v$} as the subtree rooted on $v$ where all the nodes of the subtree are labeled with their ancestor vector. Note that in this construction, the ID of the nodes do not appear, hence several nodes might have the same type while being at completely different places in the graph or the tree.

Let us now define a subgraph of $G$ that we will call the \emph{$k$-reduced graph}.
If a node $u$ has more than $k$ children of the same type, a \emph{valid pruning} operation consists in removing the subtree rooted at one of these children (including the children). 
Note that in doing so, we change the structures of the subtree of $u$ and the subtrees of its ancestors, thus we also update their types. 
A \emph{$k$-reduced graph} $H$ (that is, the kernel for this parameter $k$) of $G$ is a graph obtained from~$G$ by iteratively applying valid pruning operations on a vertex of the largest possible depth in $\mathcal{T}$ while it is possible. 
A vertex $v$ is \emph{pruned} for a valid pruning sequence if it is the root of a subtree that is pruned in the sequence. Note that there are some vertices of $G \setminus H$ that have been deleted, but that are not pruned. 

Let $G$ be a graph, and $H$ be a $k$-reduced graph of $G$.
The \emph{end type} (with respect to $H$)\footnote{One can prove that it actually does not depend on $H$ but we do not need it in our proof.} of a vertex~$v$ of~$G$ is: its type in~$H$ if it has not been deleted, and the last type it has had otherwise (that is, its type in the graph $G'$ which is the current graph when it was deleted).

\subsection{Size of the kernel and number of end types}
Since we apply pruning operations on a vertex of the largest possible depth, if at some point we remove a vertex of depth $i$, then we never remove a subtree rooted on a vertex of depth strictly larger than $i$ afterwards. It implies that when a vertex at depth $i$ is deleted, the types of the nodes at depth at least $i$ are their end type. 
The following lemma, proved in Appendix~\ref{app:same-type}, describes the structure of the end types in the $k$-reduced graph. 

\begin{restatable}{lemma}{lemSameType}
\label{lem:same-type}
Let $G$ be a graph and $H$ be a $k$-reduced graph of $G$. Let $u \notin H$ and $v \in H$, such that $u$ is a child of $v$. Then there exists exactly $k$ children of $v$ in $H$ whose end type is the end type of~$u$.
\end{restatable}

Observe that the end type of a vertex $v$ depends only on the adjacency of $v$ with its ancestors and on the number of children of $v$ of each possible end type. Combining this with Lemma~\ref{lem:same-type}, we prove the following statement.

\begin{restatable}{proposition}{propKReducedSize}
\label{prop:k-reduced-size}
The number of possible end types of a node at depth $d$ in a $k$-reduced graph of treedepth at most $t$ is bounded by $f_d(k,t):=2^d \cdot (k+1)^{f_{d+1}(k,t)}$. It follows that the size of each $k$-reduced graph only depends on $k$ and $t$.
\end{restatable}

The proof of Proposition~\ref{prop:k-reduced-size} is in Appendix~\ref{app:k-reduced-size}. 
The idea is to have a bottom-up induction. 
For the leaves of the tree, the type only depends on the adjacency of the vertex to its ancestors in the tree, therefore there are only $2^t$ types. Then, for an internal node, as there can be only $k$ children with the same type, the fact that there is a bounded number of children types implies that there is a bounded number of types for this internal vertex.

\subsection{Correctness of the kernel}

\begin{restatable}{proposition}{propKernelCorrectness}
\label{prop:kernel-correctness}
Let $G$ be a graph of treedepth $t$, $\mathcal{T}$ be a $t$-model of $G$, and $G'$ be a $k$-reduced graph of $G$. Then $G\simeq_k G'$ (using the notation of Subsection~\ref{subsec:preliminaries-logics}).
\end{restatable}

\begin{proof}
Observe that $G'$ is a subgraph of $G$, and denote by $\mathcal{T}'$ the restriction of $\mathcal{T}$ to the vertices of $G'$.
If $S\subset V(G)$, we denote by $\mathcal{T}_S$ the subtree of $\mathcal{T}$ induced by the vertices of $S$ and their ancestors. In particular, $\mathcal{T}'=\mathcal{T}_{V(G')}$. Moreover, two rooted trees are said to be \emph{equivalent} if there is an end type-preserving isomorphism between them.

By Theorem~\ref{thm:EF}, proving Proposition~\ref{prop:kernel-correctness} is equivalent to finding a winning strategy for Duplicator in the Ehrenfeucht-Fraissé game on $G,G'$ in $k$ rounds. To this end, we prove that she can play by preserving the following invariant.

\begin{claim}
Let $x_1,\ldots,x_{i}$ (resp. $y_1,\ldots,y_{i}$) be the positions played in $G$ (resp. $G'$) at the end of the $i$-th turn. Then the rooted trees $\mathcal{T}_{\{x_1,\ldots,x_i\}}$ and $\mathcal{T}'_{\{y_1,\ldots,y_i\}}$ are equivalent. 
\end{claim}

The invariant holds for $i=0$, since the two trees are empty. Assume now that it is true for some $i<k$. We consider the case where Spoiler plays on vertex $x_{i+1}$ in $G$, the other case being similar (and easier). 
Consider the shortest path in $\mathcal{T}_{\{x_1,\ldots,x_{i+1}\}}$ between $x_{i+1}$ and a vertex of $\mathcal{T}_{\{x_1,\ldots,x_{i}\}}$. We call this path $u_1,...,u_p$, with $u_1$ a node of $\mathcal{T}_{\{x_1,\ldots,x_{i}\}}$ and $u_p=x_{i+1}$. Note that, necessarily, for all $j\in [1,i]$, $u_j$ is the parent of $u_{j+1}$ in the tree. 
For $j=1,\ldots,p$, we will find a vertex $u'_j$ in $G'$ such that  $\mathcal{T}_{\{x_1,\ldots,x_{i},u_j\}}$ is equivalent to $\mathcal{T}'_{\{y_1,\ldots,y_{i},u'_j\}}$ (this implies that $u_j$ and $u'_j$ have the same end type). 

For $j=1$, first observe that $\mathcal{T}_{\{x_1,\ldots,x_{i},u_1\}}=\mathcal{T}_{\{x_1,\ldots,x_{i}\}}$, because $u_1$ belongs to $\mathcal{T}_{\{x_1,\ldots,x_{i}\}}$. Then, since $\mathcal{T}_{\{x_1,\ldots,x_{i}\}}$ is equivalent to $\mathcal{T}'_{\{y_1,\ldots,y_{i}\}}$, we can define $u'_1$ as the copy of $u_1$ in $\mathcal{T}'_{\{y_1,\ldots,y_{i}\}}$.

Assume now that $u'_1,\ldots,u'_j$ are constructed. 
Let $T$ be the end type of $u_{j+1}$ in $G$, and $r$ be the number of children of $u_j$ having $T$ as their end type (including $u_{j+1}$). 
By construction of $G'$ and $u'_j$, we know that $u'_j$ has $\min(r,k)$ children with type $T$ in $\mathcal{T'}$. 
Observe that at most $\min(r-1,i)$ children of $u_j$ of type $T$ in $\mathcal{T}$ can lie in $\mathcal{T}_{\{x_1,\ldots,x_i\}}$. 
Indeed, since $u_{j+1}$ does not belong to $\mathcal{T}_{\{x_1,\ldots,x_i\}}$, we get the $r-1$ term, and since $\mathcal{T}_{\{x_1,\ldots,x_i\}}$ is made by $i$ vertices and their ancestors, not more than $i$ vertices of $\mathcal{T}_{\{x_1,\ldots,x_i\}}$ can have the same parent.
Also, using $i<k$, we get $\min(r-1,i)\leqslant \min(r,k)-1$. Therefore, there exists a child $u'_{j+1}$ of $u'_j$ of type $T$ in $\mathcal{T}'\setminus\mathcal{T}'_{\{y_1,\ldots,y_i\}}$. 

By taking $y_{i+1}=u'_p$, we finally obtain that  $\mathcal{T}_{\{x_1,\ldots,x_{i},u_p\}}=\mathcal{T}_{\{x_1,\ldots,x_{i+1}\}}$ is equivalent to $\mathcal{T}'_{\{y_1,\ldots,y_{i},u'_p\}}=\mathcal{T}'_{\{y_1,\ldots,y_{i+1}\}}$, as required. 
\end{proof}

\subsection{Certification of the kernel}
\label{subsec:kernel-certif}

\begin{proposition}
Let $k$ be an integer. Let $G$ be a graph of treedepth at most $t$ with a coherent model $\mathcal{T}$. Let $H$ be a $k$-reduction of $G$ obtained via a valid pruning from $\mathcal{T}$.
Then we can certify with certificates of size $O(t\log n+g(k,t))$ that $H$ is a $k$-reduction of $G$ from $\mathcal{T}$.
\end{proposition}


\begin{proof}

Let us describe a local certification. On a \emph{yes}-instance, the prover gives to every vertex $v$ the following certificate:
\begin{itemize}
    \item The $O(t\log n)$-bit certificate of $v$ for the $t$-model $\mathcal{T}$ of $G$ given in Theorem~\ref{thm:certify-treedepth}.
    \item A list of $d$ booleans that says, for any ancestor $x$ of $v$, including $v$, if $x$ is pruned, i.e. the subtree rooted on $x$ has been pruned at some step.
    \item For every ancestor $w$ of $v$ including $v$, the end type of $w$, coded on $\log (f_i(k,t))$ bits, where $i$ is the depth of $w$ (by Proposition~\ref{prop:k-reduced-size}).
\end{itemize}
Every node $v$ at depth $d$ thus receives a certificate of size at most $O(t\log n+d+\sum_{i=1}^d\log (f_i(k,t)))$. Let us now describe the local verification algorithm, as well as why it is sufficient for checkability.

Recall that the end type of a vertex only depends on its adjacency with its list of ancestors as well as the end types of its children. 
So first, the node $v$ can check that its adjacency with its list of ancestors is compatible with its end type. 
Then, it checks that, if one of its children $w$ has been pruned, then it has exactly $k$ children with the type of $w$ that have not been pruned (there is no type $T$ such that more than $k$ children of type $T$ are left after pruning).
Note that $v$ has access to all this information since, for every child $w$, there is a vertex $x$ in the subtree rooted on $w$ adjacent to $v$, because $\mathcal{T}$ is coherent.
Finally, since the end type of $v$ is determined by the end types of its children, $v$ simply has to check that its end type is consistent with the list of end types of its children.

As in the proof of Theorem~\ref{thm:certify-treedepth}, for any child $w$ of $v$, if the prover has cheated and the type of $w$ has been modified between $w$ and the exit vertex of $w$, then one node of the path from $w$ to the exit vertex should discover it, which ensures that the certification is correct. 
\end{proof}

\section{Lower bounds via non-deterministic communication complexity}
\label{sec:lower-bounds}

In this section, we will prove our two lower bounds, namely Theorem~\ref{thm:isomorphism-trees} and~\ref{thm:lower-bound-treedepth}. 
To do so, we will first define a framework for reduction from two-party non-deterministic communication complexity, and then use it for the two proofs. 

Such reductions from communication complexity have been used before in local certification in~\cite{GoosS16,Censor-HillelPP20, FeuilloleyFHPP21}. 
But in all these works, the reduction was used to establish lower bounds in the polynomial regime (\emph{e.g.} $\Omega(n)$ or~$\Omega(n^2)$), whereas our second lower bound (Theorem~\ref{thm:lower-bound-treedepth}) is for the logarithmic regime. 
For both our lower bound and the lower bounds of~\cite{GoosS16,Censor-HillelPP20, FeuilloleyFHPP21}, the proof is essentially about proving that a set $S$ of vertices have to collectively know the exact structure of a far-away subgraph. The difference is that in previous works, either the subgraph was dense or the set $S$ was small, whereas in our second bound, the subgraph is sparse and the set $S$ is large, which leads to lower bounds for a lower regime.
One can naturally wonder if the other $\Omega(\log n)$ lower bounds of the area (in particular for acyclicity) can be obtained by communication complexity instead of the usual cut-and-plug techniques (that is, the combination of indistinguishability and counting arguments).

\subsection{Framework for reductions from communication complexity}
\label{subsec:framework-comcom}

\paragraph{Non-deterministic communication complexity.}

Let us describe the non-deterministic communication complexity setting. 
(This is not the same exact setting that is used in other similar reductions, we discuss the differences at the end of this subsection.)
There are two players, Alice and Bob, and a prover.
Alice has a string $s_A$ and Bob a string $s_B$. 
Both strings have length~$\ell$.
The prover chooses a string $s_P$ of length $m$, called \emph{certificate}, that is given to Alice and Bob. 
Alice decides to accept or to reject by only looking at $s_A$ and $s_P$. Let $f_A$ the function that corresponds to this process.
Same for Bob with $s_B$ and $f_B$, instead of $s_A$ and $f_A$. We say that a protocol, described by $f_A$ and $f_B$ decides \textsc{EQUALITY}, if: 
\begin{itemize}
    \item For every instance where $s_A=s_B$, there exists $s_P$ such that  $f_A(s_A,s_P)=f_B(s_B,s_P)=1$.
    \item For every instance where $s_A \neq s_B$ for all strings $s_P$, $f_A(s_A,s_P)=0$ or $f_B(s_B,s_P)=0$.
\end{itemize}

The following theorem ensures that there is asymptotically no better protocol than to have the full string written in the certificate. 

\begin{theorem}[\cite{BabaiFS86}]
\label{thm:equality}
Any non-deterministic communication protocol for \textsc{EQUALITY} for strings of length $\ell$ requires a certificate of size $\Omega(\ell)$.
\end{theorem}

\paragraph{Framework for reductions.}
Let $\ell$ be an integer. 
For any pair of strings $(s_A,s_B)$ of length $\ell$, we define a graph $G(s_A,s_B)$.

The set of vertices of $G(s_A,s_B)$ is partitioned into fours sets $V=V_A\cup V_B \cup V_\alpha \cup V_\beta$. In our reductions, the edge set of $G(s_A,s_B)$ will be composed of two parts. One will be independent of $s_A$ and $s_B$ (and will only depends on which graph class we want to obtain a lower bound and $\ell$) and a part that will depend on $s_A$ and $s_B$. 
The set of edges independent of $s_A,s_B$, denoted by $E_P$, is such that every edge in $E_P$ is in one of the following sets: $V_A \times V_{\alpha}$,  $V_\alpha\times V_\alpha$, $V_\alpha\times V_\beta$, $V_\beta\times V_\beta$ and $V_\beta\times V_B$ (see Figure~\ref{fig:comcom-1} for an illustration). 
Let $k=|V_A|=|V_B|$  and $r=|V_\alpha \cup V_\beta|$.

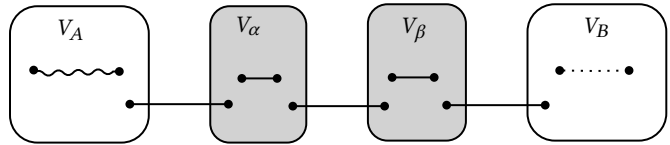
\begin{figure}[!h]
    \centering
    \tikzset{every picture/.style={line width=0.75pt}} 

\begin{tikzpicture}[x=0.75pt,y=0.75pt,yscale=-1,xscale=1]

\draw   (100,124.77) .. controls (100,117.16) and (106.16,111) .. (113.77,111) -- (156.23,111) .. controls (163.84,111) and (170,117.16) .. (170,124.77) -- (170,166.07) .. controls (170,173.67) and (163.84,179.83) .. (156.23,179.83) -- (113.77,179.83) .. controls (106.16,179.83) and (100,173.67) .. (100,166.07) -- cycle ;
\draw  [fill={rgb, 255:red, 210; green, 210; blue, 210 }  ,fill opacity=1 ] (201,120.69) .. controls (201,115.34) and (205.34,111) .. (210.69,111) -- (239.75,111) .. controls (245.1,111) and (249.43,115.34) .. (249.43,120.69) -- (249.43,170.15) .. controls (249.43,175.5) and (245.1,179.83) .. (239.75,179.83) -- (210.69,179.83) .. controls (205.34,179.83) and (201,175.5) .. (201,170.15) -- cycle ;
\draw  [fill={rgb, 255:red, 210; green, 210; blue, 210 }  ,fill opacity=1 ] (280.5,119.89) .. controls (280.5,114.43) and (284.93,110) .. (290.39,110) -- (320.05,110) .. controls (325.51,110) and (329.93,114.43) .. (329.93,119.89) -- (329.93,168.95) .. controls (329.93,174.41) and (325.51,178.83) .. (320.05,178.83) -- (290.39,178.83) .. controls (284.93,178.83) and (280.5,174.41) .. (280.5,168.95) -- cycle ;
\draw   (361,123.77) .. controls (361,116.16) and (367.16,110) .. (374.77,110) -- (417.23,110) .. controls (424.84,110) and (431,116.16) .. (431,123.77) -- (431,165.07) .. controls (431,172.67) and (424.84,178.83) .. (417.23,178.83) -- (374.77,178.83) .. controls (367.16,178.83) and (361,172.67) .. (361,165.07) -- cycle ;
\draw    (160.43,160.43) -- (208.43,160.43) ;
\draw    (240.93,161.43) -- (288.93,161.43) ;
\draw    (321.93,160.93) -- (369.93,160.93) ;
\draw  [fill={rgb, 255:red, 0; green, 0; blue, 0 }  ,fill opacity=1 ] (158.72,160.43) .. controls (158.72,159.49) and (159.49,158.72) .. (160.43,158.72) .. controls (161.38,158.72) and (162.14,159.49) .. (162.14,160.43) .. controls (162.14,161.38) and (161.38,162.14) .. (160.43,162.14) .. controls (159.49,162.14) and (158.72,161.38) .. (158.72,160.43) -- cycle ;
\draw  [fill={rgb, 255:red, 0; green, 0; blue, 0 }  ,fill opacity=1 ] (208.43,160.43) .. controls (208.43,159.49) and (209.2,158.72) .. (210.14,158.72) .. controls (211.09,158.72) and (211.85,159.49) .. (211.85,160.43) .. controls (211.85,161.38) and (211.09,162.14) .. (210.14,162.14) .. controls (209.2,162.14) and (208.43,161.38) .. (208.43,160.43) -- cycle ;
\draw  [fill={rgb, 255:red, 0; green, 0; blue, 0 }  ,fill opacity=1 ] (240.93,161.43) .. controls (240.93,160.49) and (241.7,159.72) .. (242.64,159.72) .. controls (243.59,159.72) and (244.35,160.49) .. (244.35,161.43) .. controls (244.35,162.38) and (243.59,163.14) .. (242.64,163.14) .. controls (241.7,163.14) and (240.93,162.38) .. (240.93,161.43) -- cycle ;
\draw  [fill={rgb, 255:red, 0; green, 0; blue, 0 }  ,fill opacity=1 ] (287.22,161.43) .. controls (287.22,160.49) and (287.99,159.72) .. (288.93,159.72) .. controls (289.88,159.72) and (290.64,160.49) .. (290.64,161.43) .. controls (290.64,162.38) and (289.88,163.14) .. (288.93,163.14) .. controls (287.99,163.14) and (287.22,162.38) .. (287.22,161.43) -- cycle ;
\draw  [fill={rgb, 255:red, 0; green, 0; blue, 0 }  ,fill opacity=1 ] (318.52,160.93) .. controls (318.52,159.99) and (319.28,159.22) .. (320.22,159.22) .. controls (321.17,159.22) and (321.93,159.99) .. (321.93,160.93) .. controls (321.93,161.88) and (321.17,162.64) .. (320.22,162.64) .. controls (319.28,162.64) and (318.52,161.88) .. (318.52,160.93) -- cycle ;
\draw  [fill={rgb, 255:red, 0; green, 0; blue, 0 }  ,fill opacity=1 ] (368.22,160.93) .. controls (368.22,159.99) and (368.99,159.22) .. (369.93,159.22) .. controls (370.88,159.22) and (371.64,159.99) .. (371.64,160.93) .. controls (371.64,161.88) and (370.88,162.64) .. (369.93,162.64) .. controls (368.99,162.64) and (368.22,161.88) .. (368.22,160.93) -- cycle ;
\draw    (216.93,147.43) -- (236.43,147.43) ;
\draw  [fill={rgb, 255:red, 0; green, 0; blue, 0 }  ,fill opacity=1 ] (215.22,147.43) .. controls (215.22,146.49) and (215.99,145.72) .. (216.93,145.72) .. controls (217.88,145.72) and (218.64,146.49) .. (218.64,147.43) .. controls (218.64,148.38) and (217.88,149.14) .. (216.93,149.14) .. controls (215.99,149.14) and (215.22,148.38) .. (215.22,147.43) -- cycle ;
\draw  [fill={rgb, 255:red, 0; green, 0; blue, 0 }  ,fill opacity=1 ] (233.02,147.43) .. controls (233.02,146.49) and (233.78,145.72) .. (234.72,145.72) .. controls (235.67,145.72) and (236.43,146.49) .. (236.43,147.43) .. controls (236.43,148.38) and (235.67,149.14) .. (234.72,149.14) .. controls (233.78,149.14) and (233.02,148.38) .. (233.02,147.43) -- cycle ;
\draw    (294.43,146.93) -- (313.93,146.93) ;
\draw  [fill={rgb, 255:red, 0; green, 0; blue, 0 }  ,fill opacity=1 ] (291.02,146.93) .. controls (291.02,145.99) and (291.78,145.22) .. (292.72,145.22) .. controls (293.67,145.22) and (294.43,145.99) .. (294.43,146.93) .. controls (294.43,147.88) and (293.67,148.64) .. (292.72,148.64) .. controls (291.78,148.64) and (291.02,147.88) .. (291.02,146.93) -- cycle ;
\draw  [fill={rgb, 255:red, 0; green, 0; blue, 0 }  ,fill opacity=1 ] (312.22,146.93) .. controls (312.22,145.99) and (312.99,145.22) .. (313.93,145.22) .. controls (314.88,145.22) and (315.64,145.99) .. (315.64,146.93) .. controls (315.64,147.88) and (314.88,148.64) .. (313.93,148.64) .. controls (312.99,148.64) and (312.22,147.88) .. (312.22,146.93) -- cycle ;
\draw    (111.93,144.93) .. controls (113.56,143.22) and (115.22,143.18) .. (116.93,144.81) .. controls (118.64,146.44) and (120.3,146.4) .. (121.93,144.69) .. controls (123.56,142.98) and (125.22,142.94) .. (126.93,144.57) .. controls (128.64,146.2) and (130.3,146.16) .. (131.93,144.45) .. controls (133.56,142.74) and (135.22,142.7) .. (136.93,144.33) .. controls (138.64,145.96) and (140.3,145.92) .. (141.92,144.21) .. controls (143.55,142.5) and (145.21,142.46) .. (146.92,144.09) .. controls (148.63,145.72) and (150.29,145.68) .. (151.92,143.97) -- (153.43,143.93) -- (153.43,143.93) ;
\draw  [dash pattern={on 0.84pt off 2.51pt}]  (378.93,143.43) -- (394.42,143.69) -- (413.93,143.43) ;
\draw  [fill={rgb, 255:red, 0; green, 0; blue, 0 }  ,fill opacity=1 ] (110.22,143.22) .. controls (110.22,142.28) and (110.99,141.52) .. (111.93,141.52) .. controls (112.88,141.52) and (113.64,142.28) .. (113.64,143.22) .. controls (113.64,144.17) and (112.88,144.93) .. (111.93,144.93) .. controls (110.99,144.93) and (110.22,144.17) .. (110.22,143.22) -- cycle ;
\draw  [fill={rgb, 255:red, 0; green, 0; blue, 0 }  ,fill opacity=1 ] (153.43,143.93) .. controls (153.43,142.99) and (154.2,142.22) .. (155.14,142.22) .. controls (156.09,142.22) and (156.85,142.99) .. (156.85,143.93) .. controls (156.85,144.88) and (156.09,145.64) .. (155.14,145.64) .. controls (154.2,145.64) and (153.43,144.88) .. (153.43,143.93) -- cycle ;
\draw  [fill={rgb, 255:red, 0; green, 0; blue, 0 }  ,fill opacity=1 ] (375.52,143.43) .. controls (375.52,142.49) and (376.28,141.72) .. (377.22,141.72) .. controls (378.17,141.72) and (378.93,142.49) .. (378.93,143.43) .. controls (378.93,144.38) and (378.17,145.14) .. (377.22,145.14) .. controls (376.28,145.14) and (375.52,144.38) .. (375.52,143.43) -- cycle ;
\draw  [fill={rgb, 255:red, 0; green, 0; blue, 0 }  ,fill opacity=1 ] (410.52,143.43) .. controls (410.52,142.49) and (411.28,141.72) .. (412.22,141.72) .. controls (413.17,141.72) and (413.93,142.49) .. (413.93,143.43) .. controls (413.93,144.38) and (413.17,145.14) .. (412.22,145.14) .. controls (411.28,145.14) and (410.52,144.38) .. (410.52,143.43) -- cycle ;

\draw (139.52,171.45) node [anchor=north west][inner sep=0.75pt]   [align=left] {};
\draw (122.52,116.35) node [anchor=north west][inner sep=0.75pt]    {$V_{A}$};
\draw (212.69,114.4) node [anchor=north west][inner sep=0.75pt]    {$V_{\alpha }$};
\draw (296.19,114.9) node [anchor=north west][inner sep=0.75pt]    {$V_{\beta }$};
\draw (388.19,114.9) node [anchor=north west][inner sep=0.75pt]    {$V_{B}$};

\end{tikzpicture}
    \caption{Illustration of the construction of $G(s_A,s_B)$. The straight edges are the five possible types for edges of $E_P$. The curvy edge corresponds to an edge of Alice, and the dotted edge to an edge of Bob.}
    \label{fig:comcom-1}
\end{figure}

Let $t_A$ be an injection from the set of strings of length $\ell$ to the set of subgraphs of $V_A$. 
Let $t_B$ be the analogue for $V_B$.
The graph $G(s_A,s_B)$ is the graph with vertex set $V$, and edge set $E = t_A(s_A) \cup t_B(s_B) \cup E_P$. 
Note that, by construction, the vertices of $V_A \cup V_\alpha$ are not adjacent to the ones of $V_B$, and the ones of $V_B \cup V_\beta$ are not adjacent to the ones of $V_A$.

This graph is equipped with an identifier assignment, such that the vertices of $V_\alpha \cup V_\beta$ get the identifiers from 1 to $r$ (in an arbitrary order).

\begin{restatable}{proposition}{propCCReduction}
\label{prop:CC-reduction}
Let $\mathcal{P}$ be a graph property that is satisfied by $G(s_A,s_B)$ if and only if $s_A=s_B$. 
Then a local certification for~$\mathcal{P}$ requires certificates of size $\Omega(\ell / r)$.
\end{restatable}

The proof of Proposition~\ref{prop:CC-reduction} is deferred to Appendix~\ref{app:proof-framework}. 
The idea is that Alice and Bob can use a certification in the following way. 
First, they build the graph $(V,E_P)$ that corresponds to the length $\ell$ of their strings. Then Alice adds the edges $t_A(s_A)$ on her copy, and Bob adds the edges of $t_B(s_B)$ on his copy.
Finally, they interpret the certificate given by the prover as an assignment of local certificate to the vertices of $V_\alpha$ and $V_\beta$.
They can now simulate the local verification on their part of the graph, namely the vertices of $V_A\cup V_\alpha$ and $V_B \cup V_\beta$ respectively, and thus decide if the graph has property $\mathcal{P}$ or not, which by assumption is equivalent to solve the \textsc{EQUALITY} problem. 
Now if the local certification uses certificates that are very small, it implies that the certificate used in the simulation is also small which would contradict Theorem~\ref{thm:equality}.

\paragraph{Discussion of the framework.}

Reduction to two-party non-deterministic complexity has already been used several times in local certification~\cite{GoosS16,Censor-HillelPP20, FeuilloleyFHPP21}, but for the sake of simplicity in the reduction we use a slightly different setting. 
First, we use a single certificate instead of one for each player. Second, we say that the instance is rejected if at least one player rejects, instead of having both players reject. Finally, we do not use communication between Alice and Bob: they only read the same certificate. 
It is known that these changes do not change the asymptotic complexity of the problem.

Note that the framework applies to a where the vertices can receive both a global certificate and local certificates as in \cite{FeuilloleyH18}. Also, by having $V_\alpha$ and $V_\beta$ of large enough diameter, one can derive bounds for constant-distance view, or even non-constant views (as in \cite{GoosS16, FeuilloleyFHPP21}).

\subsection{Application to fixed-point free automorphism of trees of bounded depth}

We will use the framework described in Section~\ref{subsec:framework-comcom} to prove the following theorem.

\thmIsomorphismTrees*
The same bound (without the logarithmic factors) was proved in \cite{GoosS16} for trees of unbounded depth, via a counting argument. 
Given that we have results on bounded treedepth, it is necessary to have a lower bound on bounded depth trees, to allow fair comparisons between MSO properties and non-MSO properties (\emph{e.g.} isomorphism-like properties).

The proof is deferred to Appendix~\ref{app:isomorphism}. 
It is a relatively direct use of the framework: Both $V_\alpha$ and $V_\beta$ are reduced to a single vertex connected to each other. Then $V_A$ and $V_B$ will be rooted trees whose root is connected to respectively $V_\alpha$ and $V_\beta$.
The result follows from the fact that the logarithm of the number of trees of depth $k$ is $\tilde{\Omega}(n)$, as soon as $k\geq 3$~\cite{PachPPS13}, which allows having an injection from the set of strings to the set of bounded depth trees.

\subsection{Application to treedepth certification}

\thmLowerBoundTreedepth* 
  
\begin{proof}

We first prove the theorem for $k=5$, and then explain how to modify the argument for any $k\geq 4$.
Again, we will use the framework of Subsection~\ref{subsec:framework-comcom}.
Let $\ell,n$ be two integers such that there is an injection $f$ from the set of strings of length $\ell$ to the set of matchings between two (labelled) sets of size $n$.
Our construction is illustrated in Figure~\ref{fig:comcom-2}.
Each set $V_A$, $V_B$, $V_\alpha$ and $V_\beta$ consists of two sets of $n$ vertices, that we denote with exponents, e.g. $V_A^1$ and $V_A^2$.
In each of these sets, the vertices are indexed between $1$ and $n$. 
We also add a vertex~$u$, that is adjacent to all the vertices of $V_\alpha$. In the construction, it will behave like a vertex of $V_{\alpha}$ (hence simulated by Alice). 
The set of edges $E_P$ is the collection of $2n$ disjoint paths on four nodes, of the form $(V_A^j[i],V_\alpha^j[i],V_\beta^j[i], V_B^j[i])$ for every $i \le n$ and every $j \in \{1,2\}$.
Note that the graph is connected (even without Alice and Bob's private edges), thanks to the  vertex $u$ which is complete to $V_\alpha$ and then adjacent to every path.

\begin{figure}[!h]
    \centering
    \scalebox{0.85}{
    \input{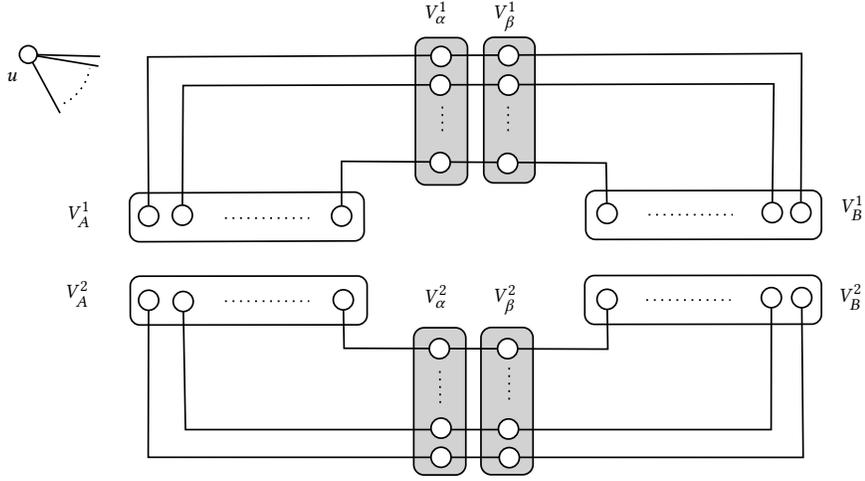}}
    \caption{Illustration of the basis of construction of $G(s_A,s_B)$ for bounded treedepth. On the picture, the upper part contains the sets $V_A^1$, $V_\alpha^1$, $V_\beta^1$, and $V_B^1$, and the lower part contains $V_A^2$, $V_\alpha^2$, $V_\beta^2$, and $V_B^2$. 
    The vertex $u$ is adjacent to all the vertices of $V_\alpha$.}
    \label{fig:comcom-2}
\end{figure}

Let us now describe the part that is private to Alice. 
Let $s_A$ be the string of length $\ell$ given to Alice and $M_A$ be the matching $f(s_A)$ between $V_A^1$ and $V_A^2$. Bob does the same for its string $s_B$.
We say that the matchings are equal if, for all $i,j$, $(V_A^1[i],V_A^2[j])$ is in Alice's matching if and only if $(V_B^1[i],V_B^2[j])$ is in Bob's matching.

\begin{restatable}{lemma}{lemTreedepthPebbles}
\label{lem:treedepth-pebbles}
If the matchings are equal, the graph has treedepth $5$, otherwise it has treedepth at least $6$.
\end{restatable}

The proof of this result can be found in Appendix~\ref{app:proof-treedepth-pebbles} and uses a cops-and-robber characterization of treedepth~\cite{GruberH08}.

Once again, we are exactly in the situation of Proposition~\ref{prop:CC-reduction}, and we want to optimize the parameters.
The number of matchings on $n$ vertices is $n!$, thus the logarithm of this quantity is of order $n\log n $.
Therefore, we can take $\ell \sim n\log n$. 
As the size of $V_\alpha \cup V_\beta$ is $2n$, by
Proposition~\ref{prop:CC-reduction} we get a $\Omega(\log n)$ lower bound.

To extend this proof to the case $k>5$, it is sufficient to remark that by adding vertices on the edges that have right corner in Figure~\ref{fig:comcom-2} (\emph{e.g.} the edges of the form $(V_A^1[i],V_\alpha^1[i]$), we can increase the length of the cycles, which changes the threshold between correct instances and incorrect instances, without changing the  the rest of the argument. 
One can actually have a proof for $k=4$, but without using in the exact framework described above, in particular removing the vertices of $V_\alpha$ and $V_\beta$, to get shorter cycles.
\end{proof}


\newpage{}

\bibliography{MSO-certif.bib}




\newpage{}
\appendix

\section{Missing proofs and discussions of Section~\ref{sec:approach}}

\subsection{Discussion of verification radius: one versus constant }
\label{app:discussion-radius}
An aspect of the model that is important in this paper is the locality of the verification algorithm. 
The original papers on local certification consider a model called \emph{proof-labeling schemes}~\cite{KormanKP10}, where the nodes only see (the certificates of) their neighbors. 
This choice originates from the state model of self-stabilization~\cite{Dolev2000}.
The model was generalized in~\cite{GoosS16} to \emph{locally checkable proofs} where the vertices can look at a constant distance. 
It is proved in~\cite{GoosS16} that the classic lower bounds (\emph{e.g.} for acyclicity) still hold in this model. 

The two models have pros and cons. Choosing constant distance is more appealing from a theoretical point of view, as it removes the distance $1$ constraint (which could seem arbitrary), but still captures a notion of locality. 
On the other hand, constant distance is not well-suited to contexts where we care about message sizes: with unbounded degree, looking at constant distance can translate into huge messages.
As noted in \cite{GoosS16}, due to their locality, FO formulas can be checked without certificate if we can adapt the view of the node to the formula, and this can be extended to certification of monadic $\Sigma_1^1$ formulas if one allows $O(\log n)$-bit certificates. 
For this paper, we chose to fix the distance to $1$, in order to prevent this adaptation of the radius to the formula. 
Note that the difference between the two models can be dramatic. For example, deciding whether a graph has diameter $3$ or more, does not need any certificate if the nodes can see at distance $3$, but requires certificates of size linear in $n$ if they can only see their neighbors~\cite{Censor-HillelPP20}.

\subsection{Proof of Lemma~\ref{lem:existsFO}: Certification of small fragments}
\label{sec:fragment}

This section is devoted to prove Lemma~\ref{lem:existsFO}.

\lemFragments*

Let us first prove the following lemma:

\begin{lemma}\label{lem:FOext}
Existential FO sentences with $k$ quantifiers (i.e. whose prenex normal form has only existential quantifiers) can be certified with $O(k \log n)$ bits.
\end{lemma}

\begin{proof}
Let $G$ be a connected graph and $\exists x_1\cdots\exists x_k\varphi$ be an existential FO sentence where $\varphi$ is quantifier-free. Let $v_1,\ldots,v_k$ be $k$ vertices such that the formula is satisfied by $v_1,\ldots,v_k$. 

Every vertex receives the following certificate:
\begin{itemize}
    \item The list of identifiers of vertices $v_1,\ldots,v_k$.
    \item The $k\times k$ adjacency matrix of the subgraph induced by $v_1,\ldots,v_k$.
    \item The certificate of a spanning tree rooted on $v_i$ for every $i \le k$ (see Subsection~\ref{subsec:basic-certif}).
\end{itemize}
Every node then checks the certificate as follows. First, every node checks that its neighbors have the same list of vertices $v_1,\ldots,v_k$ and the same adjacency matrix. Then every node checks the certificate of the spanning tree of each $v_i$. Finally, each of the vertices $v_1,\ldots,v_k$ can now use the adjacency matrix to evaluate $\varphi$ on $(v_1,\ldots,v_k)$ and check that it is satisfied.
\end{proof}

Let us now prove the second part of Lemma~\ref{lem:existsFO}.

\begin{lemma}
FO sentences with quantifier depth at most 2 can be certified with $O(\log n)$ bits.
\end{lemma}
\begin{proof}
First, observe that sentences with quantifier depth one are satisfied by either all graphs or none of them. We thus consider the depth 2 case.

Let $\varphi$ be a sentence of quantifier depth at most two. Without loss of generality, we may assume that $\varphi$ is a boolean combination of sentences of the form $Qx \psi(x)$ where $\psi(x)$ is again a boolean combination of formulas of the form $Qy\pi(x,y)$ where $\pi(x,y)$ is quantifier-free. 

Observe that up to semantic equivalence, $\pi(x,y)$ can only express that $x=y$, $xy$ is an edge, $xy$ is a non-edge, or the negation of these properties.

Trying the two possible ways of quantifying $y$ in these six properties, we end up showing (using that our graphs are connected) that $\psi(x)$ lies among these three properties or their negations:
\begin{itemize}
\item $x$ is the only vertex.
\item $x$ is a dominating vertex.
\item $x$ is not the only vertex but dominates the graph.
\end{itemize}

Now, quantifying on $x$ leaves only a few choices for $\varphi$, namely boolean combinations of the following:
\begin{enumerate}
    \item The graph has at most one vertex.
    \item The graph is a clique.
    \item The graph has a dominating vertex.
\end{enumerate}

Since certifying disjunction or conjunction of certifiable sentences without blow up (asymptotically) in size is straightforward, it is sufficient to show that the three properties and their negations can all be checked with $O(\log(n))$-bit certificates.

Since our graphs are connected, Property 1 is equivalent to say that every vertex has degree 0, which can be checked with empty certificates. Similarly, its negation is equivalent to having minimum degree 1 which can be checked similarly.

For Property 2 (resp. the negation of 3), we begin by computing the number $n$ of vertices in the graph and certifying it locally (it is well-known that this can be done with $O(\log n)$-bit certificates, see \emph{e.g.}~\cite{Feuilloley21}). 
The verification algorithm then just asks whether the degree of each vertex is $n-1$ (resp. less than $n-1$).

For Property 3 (resp. the negation of 2), we again compute and certify the number $n$ of vertices. We additionally certify a spanning tree rooted at a vertex of degree $n-1$ (resp. less than $n-1$). The root then just check that it has indeed the right degree.
\end{proof}

\subsection{Discussion of distributed graph automata}
\label{app:discussion-DGA}

In this subsection, we discuss the model of alternating distributed graph automata of~\cite{Reiter15}, which also connects MSO logic on graphs to distributed models of computation. 
This paper belongs to a series of works aiming at capturing (modal) logics on graphs with different sorts of distributed automata models, see \emph{e.g.}~\cite{HellaJKLLLSV15, Reiter17, EsparzaR20}.  

Let us quickly describe what the model of~\cite{Reiter15}, and then how it compares with our model. 
The nodes of the graph are finite-state machines, and they update their states in synchronous rounds.
There is a constant number of such  rounds. 
The nodes are anonymous, that is, the nodes are not equipped with identifiers.\footnote{In general, papers studying connections between logic and distributed computation do not use IDs, because IDs do not have a simple equivalent in logic. An exception is~\cite{BolligBR19}.}
The transition function of a node takes as input its state and the states of its neighbors in the form of a set (no counting is possible). 
At the end of the computation, the set of the states of the nodes, $F$, is considered, and the computation accepts if and only if $F$ is one of the accepting sets of states.
The alternating aspect is described in \cite{Reiter15} with computation branches, but in the context of our work it is more relevant to describe it informally as a prover/disprover game. The transition functions actually do not depend only on the states of the neighborhood, they also depend on additional labels given by two oracles, called prover and disprover. The prover and the disprover alternate in providing constant-size labels to the nodes, in order to reach respectively acceptance and rejection.

There are several substantial differences between our model and the model of~\cite{Reiter15}. First, our model is stronger in terms of local computation: we assume unbounded computation time and space whereas \cite{Reiter15} assumes finite-state machines. 
Second, our acceptance mechanism is weaker, in the sense that it is essentially the conjunction of a set of binary decisions, whereas \cite{Reiter15} uses an arbitrary function of a set of outputs. 
Third, we only have one prover, whereas \cite{Reiter15} has the full power of alternating oracles. 
Actually, variants of local certification using these two extensions have been considered (certification enhanced with general accepting functions in \cite{ArfaouiFIM14, ArfaouiFP13}, and generalized to an analogue of the polynomial hierarchy in \cite{FeuilloleyFH21, BalliuDFO18}), but here we are interested in the classic setting.
Finally, our verification is in one round, whereas the computation of~\cite{Reiter15} uses a constant number of rounds.

\section{Missing proofs of Section~\ref{sec:preliminaries}}

\subsection{Proof of Remark~\ref{rk:coherent}}
\label{app:remark-coherent}

Let us remind the remark and sketch a proof.

\rkCoherent*

One can easily remark that the following holds:

\begin{lemma}\label{lem:coherent1}
Let $G$ be a connected graph of treedepth $d$. Then there exists a tree $T$ that is a coherent $d$-model of $G$.
\end{lemma}
\begin{proof}
Let $T$ be a $d$-model of $G$ where the sum over all the vertices of $V$ of the depth of $v$ is minimized. We claim that $T$ is coherent. Assume by contradiction that there exists a vertex $v$, and one of its children $w$, such that no vertex of the subtree rooted in $w$ is connected to $v$. Let $v'$ be the lowest ancestor of $v$ connected to a vertex of $G_w$ (such a vertex must exist since $G$ is connected). We can attach the subtree of $w$ on $v'$ rather than $v$, without breaking the fact that the tree is a model of $G$. This new tree has a lower sum of depths than the original one, a contradiction with the minimality.
\end{proof}

Note that we cannot assume that $w$ is connected to its closest ancestor, for instance on the representation for a $k$-model of a path $P_{2^k-1}$ (see Fig.~\ref{fig:treedepth}).
Using Lemma~\ref{lem:coherent1}, one can easily check that the remark holds.

\section{Missing proofs and discussions of Section~\ref{sec:MSO-trees-automata}}

\subsection{Proof of Theorem~\ref{thm:MSO-trees}}
\label{app:proof-MSO-trees}

\thmMSOTrees*

\begin{proof}
We start by some preliminaries on tree automata, then describe the certification, and finally prove its correctness.

\paragraph{Preliminaries about tree automata.}

Before we describe the certification, let us note that in rooted trees, the adjacency is oriented: given two vertices $x$ and $y$, the basic predicates are: "$x$ is a child of $y$", and "$y$ is a child of $x$". 
In our (unoriented) MSO formalism, our basic predicate is ``$x$ and $y$ are adjacent''. 
Any MSO formula in our setting can be transferred to the oriented setting, by simply replacing every occurrence of $x-y$ by "$x$ is a child of $y$ or $y$ is a child of $x$". 
This transformation only induces a constant blow-up of the formula size, and works for any orientation of the tree. Therefore, we may assume that the trees we consider are rooted, have unbounded degree, unbounded depth and no ordering on the children of each node. We may also assume that the nodes of our trees are labeled (with finitely many labels). While this is not needed for our initial purposes, our proof gives this extension for free. 

Proposition~8 from~\cite{BonevaT05} states that a set of such trees is MSO definable if and only if it is recognized by a so-called \emph{unary ordering Presburger tree automaton}. 

Such an automaton is a quadruple $A = (Q, \Lambda, \delta, F)$, where $Q$ is a finite set of states, $F \subseteq Q$ is a set of accepting states, and $\Lambda$ is a set of nodes labels. The definition of the transition function $\delta$ is slightly technical, but for our purposes, we may only consider that $\delta$ associates each pair $(q,L)\in Q \times \Lambda$ with a computable function $\delta_{q,L}$ taking as input a multiset of states and outputing a boolean. (We will discuss the precise definition of $\delta$ in Appendix~\ref{app:discussion-LCL}.)

This definition should be interpreted the following way. Consider a vertex with label $L$. Denote by $q$ its state and by $S$ the multiset containing the sets of its children. This configuration is correct with respect to $\delta$, if $\delta_{q,L}(S)$ is true.

\paragraph{Description of the certification.}
On a \emph{yes}-instance, the prover will choose an arbitrary root for the tree, transform the unoriented MSO formula into an oriented one, find the corresponding UOP tree automaton $A$ given by~\cite{BonevaT05}, compute an accepting run of this automaton on the tree, and then label every vertex $u$ with:

\begin{enumerate}
    \item The distance $d(u)$ from $u$ to the root, modulo 3. 
    \item The description of $A$.
    \item The state of $u$ in the accepting run.
\end{enumerate}

The local verification algorithm on every vertex $u$ is the following:

\begin{enumerate}
    \item Check the consistency of the distances:
    \begin{itemize}
        \item Either there is a neighbor with distance $d(u)-1 \mod 3$, and all the other neighbors have distance $d(u)+1 \mod 3$.
        \item Or the distance is 0, and all the neighbors have distance 1. In this case, the vertex is the root, for the rest of the verification.
    \end{itemize}
    \item Check that the description of $A$ correspond to an automaton that correspond to the (transformed) MSO formula.
    \item Consider that the vertices with distance $d(u)+1 \mod 3$ are the children. 
    Check that the state of $u$, its label, and the states of the children correspond to a correct transition in $A$. If $u$ is the root, also check that the state is an accepting state.
\end{enumerate}

Note that in this certification, every vertex is given a constant size certificate, and only needs to see the certificates of its neighbors to perform the verification. 

\paragraph{Proof of correctness.}
It is well-known that mod 3 counters are enough to ensure a globally consistent orientation of a tree. 
The different steps of the verification ensure that every local configuration correspond to a proper configuration in an automaton $A$ that recognizes exactly the formula at hand. 
Therefore, if the verification algorithm accepts everywhere, then the formula is satisfied, and if it is satisfied, the prover can label the vertices to make the verification accept.
\end{proof}

\subsection{Discussion on generalizations of LCLs}
\label{app:discussion-LCL}

\paragraph{Discussion about generalization of LCLs.}

Let us discuss how the technique we use might be useful on a more abstract level, for the generalization of locally checkable labelings (LCLs).
Locally checkable labelings are the most studied family of problems in the LOCAL model. 
These are the problem on bounded-degree graphs whose correct outputs can be described by a list of correct neighborhoods~\cite{NaorS95}. 
A classic example is coloring, where every node can check that it has been given a color different from the ones of its neighbors.
These problems have been studied in depth, and after several recent breakthroughs they are quite well-understood. 
Generalizing LCLs beyond bounded degree is challenging because there can be an infinity of correct neighborhoods. 
We argue that the techniques we use to prove Theorem~\ref{thm:MSO-trees} can give a relevant direction for generalization. 

In our proof of Theorem~\ref{thm:MSO-trees}, we used that tree automata are powerful enough to capture MSO, and we have assumed the most general model, where the transition function is a general computable function. 
This model actually recognizes much more than MSO, for example, for any computable set of integers, we can recognize the set of stars whose degrees are in this set.
There exists a more restricted model of tree automata that recognizes exactly MSO properties (on the trees we consider). 
These are the unary ordering Presburger (UOP) tree automata~\cite{BonevaT05}, that we are going to define formally now.

We use that notations of~\cite{Kepser08} (Subsection~4.3.1, Automata Related Logics), that are more self-contained than the ones of~\cite{BonevaT05}. 
An \emph{ordering Presburger constraints} is a constraint of the following grammar:
\begin{align*}
p ::=& \  t\leq t \ |\ p \wedge p \ |\ \neg p \\
t ::=& \  y \ |\ n \ |\ t+t,
\end{align*}
where $n$ is an integer, and $y$ a free variable (that takes value in the integers). 
A \emph{unary ordering constraint} is a unary constraint where every atomic constraint is unary, that is, contains only one free variable. 

A \emph{unary ordering Presburger (UOP) tree automaton} is a quadruple $A = (Q, \Lambda, \delta, F)$, where $Q$ is a finite set of states, $F \subseteq Q$ is a set of accepting states, and $\Lambda$ is a set of nodes labels. 
Let $Y_Q$ be a set of $|Q|$ free variables, then $\delta$ maps pairs $(q,L)\in Q \times \Lambda$ to a unary ordering Presburger constraints on the set $Y_Q$.

This definition should be interpreted the following way. 
Consider a configuration with a label $L$ and a state $q$ for the parent, $y_s$ children with state $s$, for every $s$. 
This configuration is correct with respect to $\delta$, if the formula $\delta(q,L)$ is satisfied. Such a formula could be, for example: there are at least 3 children with state $q_1$, and between 1 and 4 children with state $q_2$, etc.

Now, \cite{BonevaT05} (Proposition~8) establishes that: a set of node-labeled, unbounded-degree, unbounded-depth, rooted trees with no ordering on the children is MSO definable, if and only if, it is recognized by a unary ordering Presburger tree automaton.

We suggest that the special shape of the transition function, comparing numbers of states to constants, is interesting to generalize LCL (replacing states by inputs labels). 
First, it is a natural formalism, that allows to describe easily classic problems such as coloring, maximal independent set, etc.
Second, the result of~\cite{BonevaT05} shows that it exactly captures an important type of global logical formulas, at least on trees. 
Note that similar but more general versions could also be of interest. For example, (general) Presburger tree automata is a more generic formalism motivated by the structure of XML files~\cite{SeidlSM03}, where one can compare the number of occurrences of different states (instead of just comparing them to some constants).

\section{Missing proofs from Section~\ref{sec:MSO-certif-bounded-treedepth}}

\subsection{Proof of Lemma~\ref{lem:same-type}}
\label{app:same-type}

\lemSameType*

\begin{proof}
By assumption, it cannot be more than $k$ since otherwise one of the children of $v$ would have been deleted. Moreover, since $u$ is deleted but not $v$, then $u$ is the root of a subtree we deleted while pruning $v$. In particular, $u$ has at least $k$ siblings with the same type. Now since all these siblings have the same depth as $u$, their type when $u$ is deleted is their end type. To conclude, observe that by construction, at least $k$ such siblings lie in $H$ since we delete some only if at least $k$ others remain.
\end{proof}

\subsection{Proof of Proposition~\ref{prop:k-reduced-size}}
\label{app:k-reduced-size}

\propKReducedSize*

\begin{proof}
Let us prove Proposition~\ref{prop:k-reduced-size} and define $f_d$ by backward induction on $d$. 

We start with $d=t$. Since the $t$-model has depth $t$, the tree rooted on a vertex of depth $t$ should be a single vertex graph. So the set of different possible types at depth $t$ only depends on the edges between the vertex of depth $t$ and its ancestors. There are $f_t(k,t)=2^{t}$ such types. 

Now let us assume that the conclusion holds for nodes at depth $d+1$, and let us prove it for depth $d$.  Let $u$ be a vertex of depth $d$ and $v_1,\ldots,v_r$ be its children in the elimination tree. 
Since $u$ is a vertex of a $k$-reduced graph, at most $k$ children of $u$ have the same end type $T$ and, by induction, there are at most $f_{d+1}(k,t)$ end types of nodes at depth $d+1$.
So the end type of $u$ is determined by its neighbors in its list of ancestors (which gives $2^d$ choices) and the multiset of types of its children. Since $u$ has at most $k$ children of each type, the type of $u$ can be represented as a vector of length $f_{d+1}(k,t)$, where each coordinate has an integral value between $0$ and $k$. So there are at most $f_d(k,t):=2^d \cdot (k+1)^{f_{d+1}(k,t)}$ types of nodes at depth $d$.
\end{proof}

\subsection{Proof of Corollary~\ref{coro:minors}}
\label{app:coro-minors}

\coroMinors*

\begin{proof}
It is well-known that not having a given minor is a property that is expressible in MSO. Thus, as soon as we consider a class that has bounded treedepth, we can certify $H$-minor-freeness with $O(\log n)$ bits, using Theorem~\ref{thm:main}. 
The graph that are $P_t$-minor-free are known to have treedepth at most $t$~\cite{NesetrilM12}, therefore we get the first part of the corollary.
The second part of the statement relies on the fact that every 2-connected component of a $C_t$-minor-free graph is $P_{t^2}$-minor-free. Indeed, assuming this holds, we can use the fact that a decomposition into 2-connected components can be certified with $O(\log n)$-bit certificates in minor-closed classes~\cite{BousquetFP21}, and reuse the first part of the proof to conclude for $C_t$-minor-free graphs.

Consider a 2-connected component $H$ of a $C_t$-minor-free graph. Note that since $C_t$ is $2$-connected, $H$ is $C_t$-minor-free. Assume that $H$ contains a path $P$ on $t^2$ vertices $u_1,\ldots,u_{t^2}$. For each $i$, since $u_i$ is not a cut-vertex of $H$, there must be an edge $u_ju_k$ with $j<i<k$. We denote by $(r(i),\ell(i))$ the largest such pair $(k,j)$ (by convention, $(\ell(1),r(1))=(1,k)$ where $k$ is the largest integer such that $u_1u_k$ is an edge). Observe that since $H$ is $C_t$-minor-free, we have $r(i)< \ell(i)+t$. Observe that by maximality, all the $\ell(i),r(i)$'s are pairwise disjoint. Now we reach a contradiction since $H$ contains a cycle of length at least $t$ using the edges $u_{\ell(r^{(k)}(1))}u_{r^{(k+1)}(1)}$, the subpaths of $P$ between $u_{r^{(k)}(1)}$ and $u_{\ell(r^{(k+1)}(1))}$, and the subpaths of $P$ between $u_1$, $u_{\ell(r(1))}$ and $u_{\ell(r^{(t-1)}(1))}$, $u_{\ell(r^{(t)}(1))}$.
\end{proof}

\section{Missing proofs of Section~\ref{sec:lower-bounds}}

\subsection{Proof of Proposition~\ref{prop:CC-reduction}}
\label{app:proof-framework}

\propCCReduction*

\begin{proof}

Consider a local certification for $\mathcal{P}$ using certificates of size $q$. 
We will use it to define a non-deterministic communication complexity protocol deciding equality. 

Let us start with Alice side.
Alice receives $s_A$ and builds the graph $G(s_A)$, which is the same as $G(s_A, s_B)$, except that there are no edges in between vertices of $V_B$ (Alice does not know $s_B$, thus cannot build $t_B(s_B)$).
Then she receives the certificate $s_P$ from the prover, of size $r \cdot q$. 
She divides it into $r$ pieces of size $q$, and labels the vertex with identifier $i$ with the $i$-th piece. 
Note that the labeled vertices are exactly $V_\alpha \cup V_\beta$.

Now Alice will consider all possible labelings of size $q$ of $V_A$. For each such labeling, she can run the local verifier on all the vertices of $V_A \cup V_\alpha$, because for all these vertices, she knows the adjacency and has certificates. She accepts if and only if at least one such labeling makes all the vertices $V_A \cup V_\alpha$ accept.

The behavior of Bob is exactly the same, except that we replace $V_A$ by $V_B$, $V_\alpha$ by $V_{\beta}$, $s_A$ by $s_B$ etc.

\begin{claim}
There exists a  certificate $s_P$ that makes both Alice and Bob accept in the protocol above, if and only if, there exists a certificate assignment of $G(s_A,s_B)$ that makes the local verifier accept.
\end{claim}

Suppose that there is a  certificate $s_P$ that makes both Alice and Bob accept, then this certificate defines a certificate assignment for the vertices of $V_\alpha \cup V_\beta$, and if Alice and Bob accept it means that there is a way to assign certificates to $V_A$ (respectively $V_B$) such that the vertices of $V_A \cup V_\alpha$ (respectively $V_B \cup V_\beta$) accept, and by taking the concatenation of these, we get an accepting certificate assignment for the local verifier. 
Conversely, if there exists an accepting certificate assignment, then the prover can put the corresponding  certificates on $V_\alpha \cup V_\beta$, and Alice and Bob will necessarily find the rest of a correct certification and accept.

Therefore, as the property $\mathcal{P}$ is satisfied if an only if $s_A=s_B$, by hypothesis, we get a protocol for equality. 
This protocol use a certificate of size $r \cdot q$, thus by Theorem~\ref{thm:equality}, $r\cdot q$ must be in $\Omega(\ell)$, which leads to our statement.
\end{proof}

\subsection{Proof of Theorem~\ref{thm:isomorphism-trees}}
\label{app:isomorphism}

\thmIsomorphismTrees*

\begin{proof}
The application of the framework in this case is pretty straightforward. 
Let $n$ and $k$ be parameters. 
Both $V_\alpha$ and $V_\beta$ are reduced to one vertex, respectively $\alpha$ and $\beta$, and $E_P$ is just a path of length 3: $(a,\alpha, \beta,b)$, where $a$ is in $V_A$, and $b$ is in $V_B$. 
Now, $t_A$ is an injection from strings of length $\ell$ to non-isomorphic trees of depth $k$ with $n$ vertices, rooted in $a$. Bob uses the same function $t_B=t_A$, but the trees are rooted in $b$.
As already noted in \cite{GoosS16}, the graph $G(s_A,s_B)$ has a fixed-point-free automorphism if and only if the two trees are equal. 
This happens only if and only if the strings are equal, hence we can use Proposition~\ref{prop:CC-reduction}. 

Now, let us establish the lower we get from this construction. 
It is proved in~\cite{PachPPS13}, that the logarithm of the number of non-isomorphic trees on $n$ vertices of depth $k\geq 3$ is asymptotically:
\[ \frac{\pi^2}{6}\frac{n}{\log \log \cdots \log n},\]
where the denominator has $k-2$ logs. 
Therefore, up to logarithmic terms, we can take $\ell$ and $n$ of the same order of magnitude, and as $r$ is constant, we get that the certificates need to be at least linear in the size of the graph.
\end{proof}

Note that the theorem of~\cite{PachPPS13} needs $k\geq 3$. We can extend the result to $k\geq 2$ with a bound of $\Omega(\sqrt{n})$, by noting that rooted trees of depth 2 with $n$ leaves are in bijection with the integer partitions of $n$ (because grouping the leaves by parent defines a partition) and that there are order of $2^{O(\sqrt{n})}$ partitions of $n$.

\subsection{Proof of Lemma~\ref{lem:treedepth-pebbles}}
\label{app:proof-treedepth-pebbles}

\lemTreedepthPebbles*

\begin{proof}
Let us first consider the graph without the vertex $u$. 
In any case, this graph is 2-regular, thus it is a disjoint union of cycles.
If the matchings are equal, the graph is a union of $n$ cycles of length 8.
If the matchings are not equal, there is necessarily a cycle of length $16$ or larger, that goes at least twice through each set of vertex.
We show that in the first case the treedepth is at most 5, and that in the second case it is at least 6. 

To do so, we will use the following cops-and-robber characterization of treedepth\cite{GruberH08}. 
Immobile cops are placed at vertices of the graph one by one, and a robber tries to escape. More precisely, the robber chooses a position to start, and then iteratively, the following happens: the position of the future new cop is announced the robber can move to any vertex that is accessible without using the position of a cop already in place then the new cop is placed. 
The game is over when a cop is added on the robber position, and the robber cannot move.
The tree-depth is exactly the optimal number of cops needed to catch the robber.

For both cases (all cycles of length 8, or at least one of length 16 or larger), one strategy illustrated in Figure~\ref{fig:pebble-game}. It consists in first putting a cop on the vertex $u$, then two cops on opposed vertices in the cycle of the cycle whee the robber is, and then to finish with a binary search on the remaining path. 
In the case of cycles of length 8, the number of cops used is 5, and for in the other case it is strictly larger, as the robber can use the larger cycle, and one more cop will be need in the final binary search.

\begin{figure}[!h]
    \centering
    \input{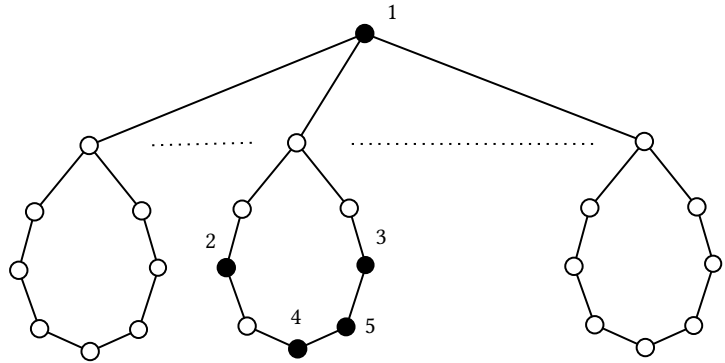}
    \caption{Illustration of the pebble game on the instance with cycles of length $8$. 
    The black vertices are the ones on which cops have been placed, and the integer indicated in which order.
    The first cop is placed at the top vertex, then the robber has to choose one of the cycles. A second and third cop are placed on two opposed vertices of the 8-cycle, at the point the robber can only move on a path of 3 vertices. A fourth vertex is placed on the vertex in the middle, and the last cop is necessarily at the place where the robber is blocked.  }
    \label{fig:pebble-game}
\end{figure}

This strategy is optimal. Indeed, as long as the vertex $u$ is not used by a cop, the robber can freely move between the cycles, and once it is chosen we are back to the beginning of the sequence described above, thus it is optimal to play it right at the beginning. The rest of the strategy is well-known to be optimal (see \emph{e.g.} \cite{NesetrilM12}).
\end{proof}

\end{document}